\documentclass{article}

\usepackage[totalwidth=450pt,totalheight=640pt]{geometry}

\usepackage{graphicx}
\usepackage{color}
\usepackage{amsmath}
\usepackage{amssymb}
\usepackage{amsfonts}
\usepackage{amsthm}
\usepackage{microtype}

\newtheorem{proposition}{Proposition}

\newtheorem{lemma}{Lemma}
\newtheorem{claim}{Claim}
\newtheorem{example}{Example}

\newtheorem{theorem}{Theorem}
\newtheorem{definition}{Definition}

\newcommand{\var}{\mathsf{var}}

\newcommand{\fw}{\mathsf{fw}}
\newcommand{\ctw}{\mathsf{ctw}}
\newcommand{\sdw}{\mathsf{sdw}}
\newcommand{\tw}{\mathsf{tw}}
\newcommand{\fiw}{\mathsf{fiw}}
\newcommand{\sd}{\mathsf{sd}}

\begin{document}

\title{Circuit Treewidth, Sentential Decision, and Query Compilation}

\author{
  Simone Bova\\TU Wien\\ \texttt{simone.bova@tuwien.ac.at}
  \and
  Stefan Szeider\\TU Wien\\ \texttt{sz@ac.tuwien.ac.at}
}

\date{}

\maketitle

\begin{abstract}
The evaluation of a query over a probabilistic database 
boils down to computing the probability of a suitable Boolean function, 
the lineage of the query over the da\-ta\-ba\-se.  
The method of query compilation approaches the task  
in two stages: first, the query lineage 
is implemented (compiled) in a  
circuit form where probability computation is tractable;
and second, the desired probability is computed over the compiled circuit.  
A basic theoretical quest in query compilation is 
that of identifying pertinent classes of queries whose lineages 
admit compact representations over increasingly succinct, 
tractable circuit classes. 

Fostering previous work by Jha and Suciu \cite{DBLP:conf/icdt/JhaS12} 
and Petke and Razgon \cite{DBLP:conf/sat/RazgonP13}, we focus on queries whose lineages 
admit circuit implementations with small tree\-wid\-th, 
and investigate their compilability within tame classes 
of decision diagrams.  In perfect analogy with the characterization of 
bounded circuit pathwidth by bounded OBDD width \cite{DBLP:conf/icdt/JhaS12}, we show that a class of 
Boolean functions has bounded circuit treewidth if and only if 
it has bounded SDD width.  Sentential decision diagrams (SDDs) are central in knowledge compilation, 
being essentially as tractable as OBDDs \cite{DBLP:conf/ijcai/Darwiche11} 
but exponentially more succinct \cite{DBLP:conf/aaai/Bova16}.  By incorporating constant width 
SDDs and polynomial size SDDs, 
we refine the panorama of query compilation 
for unions of conjunctive queries with and without inequalities \cite{DBLP:conf/icdt/JhaS11,DBLP:conf/icdt/JhaS12}.
\end{abstract}

\section{Introduction}

A basic problem in database theory 
is \emph{query evaluation in probabilistic databases}: 
Given a (Boolean) query $Q$ and a probabilistic database $D$, 
where each tuple has a given probability, 
compute the probability of the lineage of $Q$ over $D$. 
The problem is computationally 
hard, even for fixed queries in simple syntactic forms \cite{DBLP:conf/pods/GradelGH98,DBLP:journals/vldb/DalviS07}.

The \emph{lineage of a (Boolean) query $Q$ over a database $D$} 
is a monotone Boolean function $L(Q,D)$ over the tuples in $D$  
that 
accepts a subset $D'$ of tuples of $D$ 
if and only if $Q$ is true in $D'$.  A standard approach to probabilistic query evaluation 
is \emph{query compilation} \cite[Chapter~5]{Suciu:2011:PD:2031527}. Here, 
the lineage $L(Q,D)$, given as a Boolean circuit, 
which is computable in polynomial time if $Q$ is fixed, 
is implemented within a succinct circuit class 
where its probability is efficiently computable. 
In other words, to avoid computing the probability of a circuit, 
which is hard, the circuit is first \emph{compiled} to a tamer form.

Circuit classes supporting tractable probability computation, 
and model counting in particular,  
are central in \emph{knowledge compilation} \cite{DBLP:journals/jair/DarwicheM02}; 
particular emphasis is posed on the hierarchy 
of deterministic decomposable circuits.  A circuit is \emph{decomposable} 
if its AND gates represent independent probabilistic events \cite{DBLP:journals/jacm/Darwiche01}, 
and \emph{deterministic} if its OR gates represent exclusive probabilistic events \cite{DBLP:journals/jancl/Darwiche01}: 
probability computation is then feasible in linear 
time on deterministic decomposable circuits.

Aimed at a detailed syntactic classification of tractable cases 
of probabilistic query evaluation, 
Jha and Suciu amply explored the compilability 
of various classes of queries 
into various classes of deterministic decomposable circuits, 
fruitfully bridging database theory and knowledge compilation \cite{DBLP:conf/icdt/JhaS11,DBLP:journals/mst/JhaS13}.  
In this context, they studied the compilability of queries whose lineages have small circuit treewidth 
into decision diagrams, OBDDs in particular \cite{DBLP:conf/icdt/JhaS12}; 
a study 
we continue in this article. 

An \emph{ordered binary decision diagram} (OBDD)  
is a deterministic read-once branching program where every path from 
the root to a leaf visits the Boolean variables in the same order \cite{DBLP:journals/tc/Bryant86}.\footnote{OBDDs are 
deterministic decomposable circuits.}   
The OBDD size of a Boolean function is the size (number of nodes) of its smallest OBDD 
implementation.  The width of an OBDD is the largest number of nodes 
labeled by the same variable, and the OBDD width of a Boolean function 
is the smallest width attained by its OBDD implementations.  

Unifying several known tractable cases of the probability computation problem,  
Jha and Suciu introduce a structural parameter for Boolean functions, 
called \emph{expression width} \cite{DBLP:conf/icdt/JhaS12} or \emph{circuit treewidth} \cite{DBLP:conf/sat/RazgonP13}, 
that measures, for any Boolean function, the smallest treewidth of a circuit computing the function.  
They show that a Boolean function of $n$ variables and circuit treewidth $k$ 
has OBDD size 
\begin{equation}\label{eq:intr-1}
n^{O(f(k))} 
\end{equation}
where $f$ is a fast growing (double exponential) function \cite{DBLP:conf/icdt/JhaS12}.  The bound is tight 
in the sense that there are Boolean functions of $n$ variables and circuit treewidth $k$ 
whose OBDD size is $n^{\Omega(k)}$ \cite{DBLP:conf/kr/Razgon14}.\footnote{The lower bound holds even for primal treewidth, 
which is (unboundedly) larger than circuit treewidth.} 

The bound (\ref{eq:intr-1}) gives polynomial size OBDD implementations for circuits of bounded treewidth, 
but the degree of the polynomial depends (badly) on the treewidth.  
However, as Jha and Suciu show \cite{DBLP:conf/icdt/JhaS12}, 
restricting to functions of small \emph{circuit pathwidth} resolves this issue.   
Indeed, a Boolean function of $n$ variables and circuit pathwidth $k$ 
has OBDD width $f(k)$, hence OBDD size
\begin{equation}\label{eq:intr-2}
O(f(k)n)\text{;}
\end{equation}
and conversely, every Boolean function of OBDD width $k$ has circuit pathwidth $O(k)$.  Therefore, 
\emph{a class of Boolean functions has bounded circuit 
pathwidth if and only if it has bounded OBDD width.}

As Jha and Suciu conclude, the quest naturally arises 
for a similar characterization of bounded circuit tree\-wi\-d\-th.  
The quest involves, for starters, identifying a circuit class ideally as tractable as OBDDs but more succinct, 
and therefore capable of matching the bound (\ref{eq:intr-2}).  Natural candidates, 
like FBDDs or even nondeterministic read-once branching programs fail  
\cite{DBLP:journals/algorithmica/Razgon16}.

The question is natural 
and nontrivial. Compared to the substantial understanding of the compilability 
of CNF circuits parameterized by treewidth, or even cliquewidth \cite{DBLP:conf/sat/BovaCMS15,DBLP:conf/sat/Mengel16}, 
the parameterized compilability of general circuits is relatively unexplored 
and poorly understood; which is unsatisfactory because, in theory, 
the circuit treewidth of a class of Boolean functions 
can be bounded on general circuits and unbounded on CNFs \cite[Example~2.9]{DBLP:conf/icdt/JhaS12}; 
and in practice, query lineages are often presented by circuits rather than by CNFs \cite{DBLP:conf/pods/GreenKT07}.  

The only bound on the size of a compilation for a circuit 
that avoids a dependence on its treewidth in the exponent,\footnote{A special case where such a compilation is available 
is that of lineages of an MSO query over databases of bounded treewidth, 
which have linear size deterministic decomposable forms \cite{DBLP:conf/icalp/AmarilliBS15,DBLP:conf/pods/AmarilliBS16}.} like in (\ref{eq:intr-2}) as opposed to (\ref{eq:intr-1}), is a compilation of \emph{size} $m$ circuits 
into \emph{decomposable} forms of size 
\begin{equation}\label{eq:intr-3}
O(g(k)m)
\end{equation}
by Petke and Razgon \cite{DBLP:conf/sat/RazgonP13}, where $g$ is an exponential function.  
This compilation, however, lacks two features that are either needed by or desirable in its 
intended application to query compilation. The crucial missing feature 
is that decomposable circuits, in the absence of determinism, do not support model counting 
(nor, then, probability computation).  
Besides, in the upper bound (\ref{eq:intr-3}), the size of the compilation depends on the size 
of the circuit, $m$, not just on the number of its variables, $n$, 
and the former can be much larger than the latter.  Indeed, 
Petke and Razgon ask whether decomposable forms of size linear in $n$  
are attainable for circuits of bounded treewidth \cite[Section~5]{DBLP:conf/sat/RazgonP13}.

\subsection*{Contribution}

In the first part of the article (Section~\ref{sect:part1}), we show that \emph{a class of Boolean functions has bounded circuit 
treewidth if and only if it has bounded SDD width}, 
which perfectly complements the aforementioned characterization of 
circuit pathwidth via OBDD width by Jha and Suciu. 

More precisely, we prove the following (Theorem~\ref{th:linear-sdd-size} and surrounding discussion).
\begin{description}
\item[Result~1.] \textit{A Boolean circuit of $n$ variables 
and treewidth $k$ has SDD width $f(k)$, 
thus SDD size
\begin{equation}\label{eq:int-4}
O(f(k)n)\text{,} 
\end{equation}
where $f$ is a triple exponential function.  
Conversely, every Boolean function of SDD width $k$ has circuit treewidth $O(k)$.}  
\end{description}

Introduced by Darwiche \cite{DBLP:conf/ijcai/Darwiche11}, \emph{sentential decision diagrams} (SDDs) 
are a relaxation of OBDDs based on a generalized form of Shannon decomposition.  
An OBDD respecting the variable ordering $x_1<x_2<\cdots<x_n$ 
takes a \emph{binary decision} of the form 
$(x_1 \wedge S_1 ) \vee (\neg x_1  \wedge S_2)$, 
where $S_1$ and $S_2$ are OBDDs respecting the variable ordering $x_2<\cdots<x_n$.  
Intuitively, based on a binary case distinction on $x_1$, 
the OBDD executes subOBDDs respecting the subordering $x_2<\cdots<x_n$.  
An SDD respecting the variable tree $T$, whose left and right subtrees split the variables in two disjoint blocks $X$ and $Y$, 
takes a \emph{sentential decision} of the form $\bigvee_{i=1}^m (P_i(X) \wedge S_i(Y) )$. 
Here, based on an $m$-ary (exhaustive and disjoint) 
case distinction on $X$, 
implemented by the $m$ SDDs $P_i(X)$ respecting the left subtree of $T$, 
the $m$ SDDs $S_i(Y)$ respecting the right subtree of $T$ are executed.

SDDs are theoretically very robust, 
being essentially as tractable as OBDDs \cite{DBLP:conf/ijcai/Darwiche11,DBLP:conf/aaai/BroeckD15} but 
exponentially more succinct \cite{DBLP:conf/aaai/Bova16}.  
They are also appealing in practice since 
an SDD compiler is reasonable to design and implement 
(as opposed to an FBDD compiler for instance, 
whose design is already fairly elusive). 
Indeed, available SDDs compilers already yield 
more succinct SDDs than OBDDs, leveraging the 
additional flexibility offered by variable trees compared to variable orders \cite{DBLP:conf/aaai/ChoiD13,DBLP:conf/ijcai/OztokD15}.  
Moreover, SDDs have canonical forms, and hence carry a natural notion of width which in particular implies, 
if bounded, linear size implementations, exactly as OBDD width does for OBDDs \cite{DBLP:conf/ijcai/Darwiche11}.  
Thus, quite remarkably, our study unveils that circuit treewidth is characterized by an 
independently introduced, theoretically solid, 
and practically useful notion of circuit width, namely, SDD width.\footnote{SDDs were not even a natural candidate as, 
until recently \cite{DBLP:conf/aaai/Bova16}, they were conjectured to be quasipolynomially simulated by OBDDs 
(personal communication with Vincent Liew).}  

\begin{figure}[t]
\begin{center}
\begin{picture}(0,0)%
\includegraphics{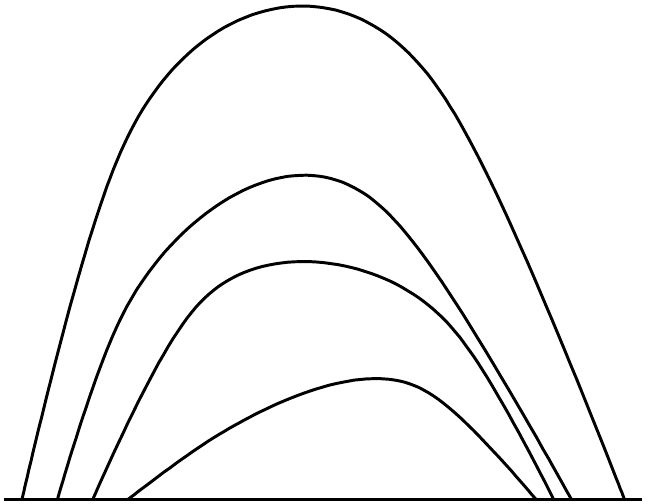}%
\end{picture}%
\setlength{\unitlength}{3729sp}%
\begingroup\makeatletter\ifx\SetFigFont\undefined%
\gdef\SetFigFont#1#2#3#4#5{%
  \reset@font\fontsize{#1}{#2pt}%
  \fontfamily{#3}\fontseries{#4}\fontshape{#5}%
  \selectfont}%
\fi\endgroup%
\begin{picture}(3284,2549)(6189,-10343)
\put(7164,-9063){\makebox(0,0)[lb]{\smash{{\SetFigFont{6}{7.2}{\familydefault}{\mddefault}{\updefault}{\color[rgb]{0,0,0}$\mathrm{OBDD}(n^{O(1)})$}%
}}}}
\put(7793,-8476){\makebox(0,0)[lb]{\smash{{\SetFigFont{6}{7.2}{\familydefault}{\mddefault}{\updefault}{\color[rgb]{0,0,0}$\mathrm{SDD}(n^{O(1)})$}%
}}}}
\put(7111,-10231){\makebox(0,0)[lb]{\smash{{\SetFigFont{6}{7.2}{\familydefault}{\mddefault}{\updefault}{\color[rgb]{0,0,0}$\mathrm{OBDD}(O(1))=\mathrm{CPW}(O(1))$}%
}}}}
\put(7059,-9650){\makebox(0,0)[lb]{\smash{{\SetFigFont{6}{7.2}{\familydefault}{\mddefault}{\updefault}{\color[rgb]{0,0,0}$\mathrm{SDD}(O(1))=\mathrm{CTW}(O(1))$}%
}}}}
\end{picture}%
\end{center}
\caption{Boolean functions.}
\label{fig:BF}
\end{figure}
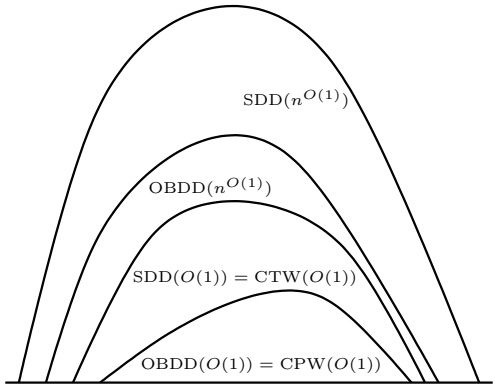 

Figure~\ref{fig:BF} illustrates the compilability panorama for Boolean 
functions relative to bounded circuit pathwi\-dth/OBDD width, 
bounded circuit treewidth/SDD wi\-dth, and polynomial OBDD and SDD size.  
The class $\mathrm{OBDD}(f(n))$ contains all Boolean functions 
of OBDD wi\-dth $f(n)$, and similarly for SDDs; the class $\mathrm{CTW}(f(n))$ 
contains all Boolean functions of circuit treewidth $f(n)$, 
and similarly for CPW and circuit pathwidth.  We have 
\begin{align*}
\mathrm{CPW}(O(1)) &= \mathrm{OBDD}(O(1)) & \text{\cite{DBLP:conf/icdt/JhaS12}}\\
                   &\subsetneq \mathrm{CTW}(O(1))& \text{\cite{DBLP:conf/icdt/JhaS12}}\\
                   & =\mathrm{SDD}(O(1))& \text{Result~1}\\
                   & \subsetneq \mathrm{OBDD}(n^{O(1)})& \text{\cite{DBLP:conf/icdt/JhaS12}}\\
                   & \subsetneq \mathrm{SDD}(n^{O(1)})& \text{\cite{DBLP:conf/aaai/Bova16}}
\end{align*}

Jha and Suciu leave open the question whether the circuit treewidth of a Boolean function is computable \cite[Section~6]{DBLP:conf/icdt/JhaS12}.  
Using the fact that satisfiability of MSO-sentences is decidable on graphs of bounded treewidth \cite{DBLP:journals/apal/Seese91}, 
we answer the question positively (Proposition~\ref{prop:seese}).

\begin{description}
\item[Result~2.] \textit{The circuit treewidth of a Boolean function is computable.}  
\end{description}

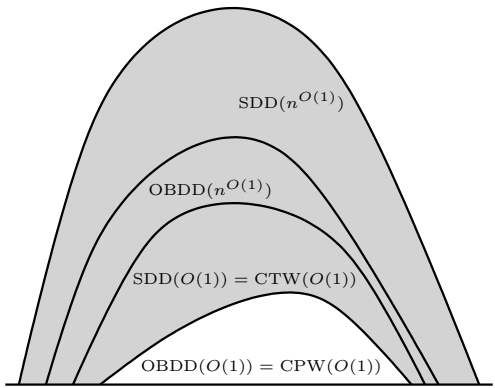
\begin{figure}[t]
\begin{center}
\begin{picture}(0,0)%
\includegraphics{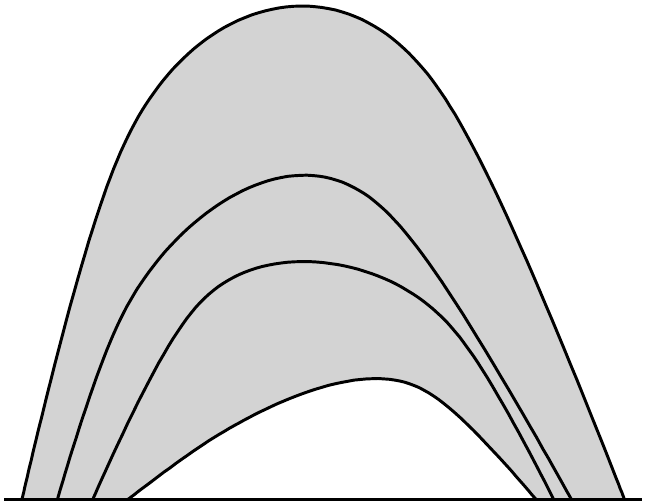}%
\end{picture}%
\setlength{\unitlength}{3729sp}%
\begingroup\makeatletter\ifx\SetFigFont\undefined%
\gdef\SetFigFont#1#2#3#4#5{%
  \reset@font\fontsize{#1}{#2pt}%
  \fontfamily{#3}\fontseries{#4}\fontshape{#5}%
  \selectfont}%
\fi\endgroup%
\begin{picture}(3284,2549)(6189,-10343)
\put(7762,-8479){\makebox(0,0)[lb]{\smash{{\SetFigFont{6}{7.2}{\familydefault}{\mddefault}{\updefault}{\color[rgb]{0,0,0}$\mathrm{SDD}(n^{O(1)})$}%
}}}}
\put(7059,-9650){\makebox(0,0)[lb]{\smash{{\SetFigFont{6}{7.2}{\familydefault}{\mddefault}{\updefault}{\color[rgb]{0,0,0}$\mathrm{SDD}(O(1))=\mathrm{CTW}(O(1))$}%
}}}}
\put(7111,-10231){\makebox(0,0)[lb]{\smash{{\SetFigFont{6}{7.2}{\familydefault}{\mddefault}{\updefault}{\color[rgb]{0,0,0}$\mathrm{OBDD}(O(1))=\mathrm{CPW}(O(1))$}%
}}}}
\put(7164,-9063){\makebox(0,0)[lb]{\smash{{\SetFigFont{6}{7.2}{\familydefault}{\mddefault}{\updefault}{\color[rgb]{0,0,0}$\mathrm{OBDD}(n^{O(1)})$}%
}}}}
\end{picture}%
\end{center}
\caption{Lineages of UCQs. The gray region is empty.}
\label{fig:ucq}
\end{figure} 

\medskip\noindent In the second part of the article (Section~\ref{sect:part2}), we study the implications of our compilability results in query compilation, 
refining the picture drawn by Jha and Suciu for unions of conjunctive queries (UCQs) with and without inequalities \cite{DBLP:conf/icdt/JhaS11,DBLP:conf/icdt/JhaS12}.  

We prove the following statement (Theorem~\ref{th:main-ineq}).

\begin{description}
\item[Result~3.] \textit{A union of conjunctive queries  
with or without inequalities containing inversions 
has lineages of exponential deterministic structured size.}  
\end{description}

Introduced by Dalvi and Suciu \cite{DBLP:conf/pods/DalviS07a}, \emph{inversion freeness} is a syntactic 
property of UCQs and UCQs with inequalities that implies compilability of their lineages in 
constant width (linear size) OBDDs and polynomial size OBDDs \cite{DBLP:conf/icdt/JhaS11,DBLP:conf/icdt/JhaS12}.  
On the other hand, if a query contains inversions, 
then it has lineages with large OBDD \cite{DBLP:conf/icdt/JhaS11,DBLP:conf/icdt/JhaS12}, 
and even SDD \cite{DBLP:conf/uai/BeameL15}, implementations.

\emph{Structuredness} is a strong form of decomposability where not only 
for every AND gate the circuits leading into the gate are defined on disjoint sets of variables \cite{DBLP:journals/jacm/Darwiche01}, 
but their variables are partitioned accordingly to an underlying variable tree \cite{DBLP:conf/aaai/PipatsrisawatD08}.  

As alluded in their informal description above, 
SDDs as well as OBDDs are special deterministic structured forms.  
Therefore Result~3 formally generalizes analogous 
previous incompilability results for SDDs and OBDDs \cite{DBLP:conf/uai/BeameL15, DBLP:conf/icdt/JhaS11}.  
The proof has the main and 
sole merit to 
combine proof ideas of Jha and Suciu together 
with lower bound techniques for deterministic structured circuits 
based on single partition communication complexity \cite{DBLP:conf/uai/BeameL15,BovaEtAlIJCA16}. 

A careful inspection of the proof shows that  
Result~3 also exponentially separates disjunctive normal forms (DNFs), 
and even prime implicant forms (IPs), from structured deterministic negation normal forms (NNFs).  
In this interpretation Result~3 settles a special case of the much harder problem 
of separating DNFs (and IPs) and deterministic decomposable NNFs (d-DNNFs); 
which, 
thanks to the recently established separation of decomposable NNFs (DNNFs) and d-DNNFs \cite{BovaEtAlIJCA16}, 
is the last open question about the relative succinctness of the compilation languages 
considered in the classic article by Darwiche and Marquis \cite{DBLP:journals/jair/DarwicheM02}.

Figure~\ref{fig:ucq} and Figure~\ref{fig:ucqneq} give an overview of 
query compilability for lineages of UCQs with and without inequalities.  For lineages of UCQs we have that
\begin{align*}
\mathrm{OBDD}(O(1)) &= \mathrm{SDD}(O(1)) \\
                   &= \mathrm{OBDD}(n^{O(1)})=\mathrm{SDD}(n^{O(1)})
\end{align*}
because by Result~3 inversions imply large structured deterministic forms, 
hence large SDDs; on the other hand, inversion freeness implies constant width OBDDs \cite{DBLP:conf/icdt/JhaS11}, 
so that $\mathrm{SDD}(n^{O(1)})\setminus \mathrm{OBDD}(O(1))=\emptyset$.  

The picture for lineages of UCQs with inequalities is
\begin{align*}
\mathrm{OBDD}(O(1)) &\subseteq \mathrm{SDD}(O(1))\\ 
                   &\subsetneq \mathrm{OBDD}(n^{O(1)}) =\mathrm{SDD}(n^{O(1)})
\end{align*}
because again Result~3 implies large SDDs in the presence of inversions, 
and inversion freeness implies polynomial size OBDDs \cite{DBLP:conf/icdt/JhaS12};
thus $\mathrm{SDD}(n^{O(1)})\setminus \mathrm{OBDD}(n^{O(1)})=\emptyset$. 
In a symmetric fashion, Jha and Suciu conjecture that, 
for lineages of UCQs with inequalities, 
it also holds that $\mathrm{SDD}(O(1))\setminus \mathrm{OBDD}(O(1))=\emptyset$.

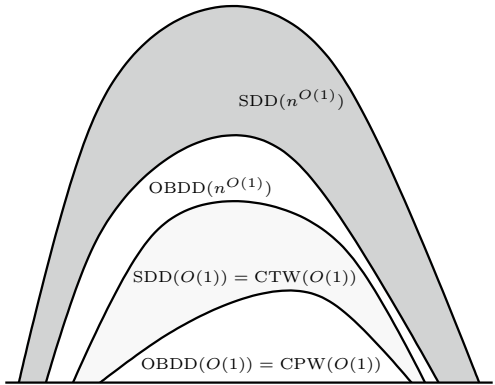
\begin{figure}[t]
\begin{center}
\begin{picture}(0,0)%
\includegraphics{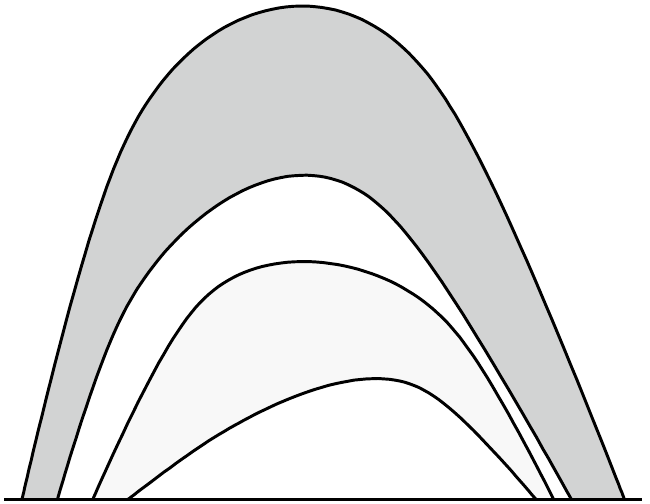}%
\end{picture}%
\setlength{\unitlength}{3729sp}%
\begingroup\makeatletter\ifx\SetFigFont\undefined%
\gdef\SetFigFont#1#2#3#4#5{%
  \reset@font\fontsize{#1}{#2pt}%
  \fontfamily{#3}\fontseries{#4}\fontshape{#5}%
  \selectfont}%
\fi\endgroup%
\begin{picture}(3284,2549)(6189,-10343)
\put(7762,-8479){\makebox(0,0)[lb]{\smash{{\SetFigFont{6}{7.2}{\familydefault}{\mddefault}{\updefault}{\color[rgb]{0,0,0}$\mathrm{SDD}(n^{O(1)})$}%
}}}}
\put(7164,-9063){\makebox(0,0)[lb]{\smash{{\SetFigFont{6}{7.2}{\familydefault}{\mddefault}{\updefault}{\color[rgb]{0,0,0}$\mathrm{OBDD}(n^{O(1)})$}%
}}}}
\put(7111,-10231){\makebox(0,0)[lb]{\smash{{\SetFigFont{6}{7.2}{\familydefault}{\mddefault}{\updefault}{\color[rgb]{0,0,0}$\mathrm{OBDD}(O(1))=\mathrm{CPW}(O(1))$}%
}}}}
\put(7059,-9650){\makebox(0,0)[lb]{\smash{{\SetFigFont{6}{7.2}{\familydefault}{\mddefault}{\updefault}{\color[rgb]{0,0,0}$\mathrm{SDD}(O(1))=\mathrm{CTW}(O(1))$}%
}}}}
\end{picture}%
\end{center}
\caption{Lineages of UCQs with inequalities. The gray region is empty, 
and the light gray region is conjectured empty.}
\label{fig:ucqneq}
\end{figure}

\subsection*{Discussion}

Our bound (\ref{eq:int-4}) amounts to a vast improvement of the available bounds.  Compared to (\ref{eq:intr-1}), 
it attains linear size versus (large degree) polynomial size compilation on bounded circuit treewidth classes.  
Compared to (\ref{eq:intr-3}), it answers abundantly 
the quest for linear size decomposable forms for Boolean functions of bounded treewidth: 
our forms are not just linear size and decomposable, 
but even deterministic and structured.  

Pushing the dependency of the compilation size down from $m$ (the size of the given circuit, as in (\ref{eq:intr-3})) to $n$ (the number of its inputs, as in (\ref{eq:int-4})) 
required an entirely new compilation idea.\footnote{This aspect of the bound, 
at first sight pedantic, is indeed relevant in query compilation, where 
the number $n$ of Boolean variables of the query lineage $L(Q,D)$ is the 
(large) number of tuples in the database $D$, and $m$ is the size of the circuit implementation of $L(Q,D)$.  
Roughly $m=O(n^{q})$, where $q$ is the size of the query $Q$. 
Hence, avoiding a dependence on $m$ means obtaining a bound where the degree of the polynomial is a universal constant, 
not just independent of the circuit treewidth of $L(Q,D)$, but also independent of $Q$.}

The idea used by Petke and Razgon \cite{DBLP:conf/sat/RazgonP13} to obtain (\ref{eq:intr-3}) 
was the following. Given a circuit $C(X)$ of $n=|X|$ variables and $m=|Z|$ gates, 
first compute its Tseitin CNF $T(X,Z)$; the circuit treewidth of the latter is (linearly) related to the former.  
To obtain a decomposable form for $C(X)$, 
existentially quantify the (gate) variables $Z$ in a decomposable form $D_T(X,Z)$ for $T(X,Z)$:
$$C(X) \equiv (\exists Z)D_T(X,Z)$$
This indirect approach via Tseitin forms, however, 
introduces two critical issues.  On the one hand, the size of $D_T(X,Z)$ depends on $|Z|=m$, 
so that the size of the compiled form will depend on $m$ 
(as opposed to depending on $n$ only).  On the other hand, for 
$(\exists Z)D_T(X,Z)$ to have size polynomial in that of $D_T(X,Z)$, 
the latter cannot be deterministic \cite{DBLP:journals/jair/DarwicheM02}; 
hence the resulting compiled form will not be deterministic either.

Our compilation approach avoids Tseitin forms and 
compiles the circuit directly; it relies on a new insight on Boolean functions, 
and more specifically on the combinatorics of their \emph{subfunctions} or \emph{cofactors} \cite{WegenerBook}. 

Dually to the notion of cofactor, 
we introduce (Definition~\ref{def:factor}) the notion of \emph{factor} of a Boolean function $F(Y,Y')$, 
that is, a function $G(Y)$ 
whose models correspond exactly to the assignments of $Y$ 
generating some cofactor of $F$.  
We then show (Lemma~\ref{lemma:factors}) that first, the rectangle $R$ formed by multiplying any two factors $G(Y)$ and $G'(Y')$ 
of $F(Y,Y')$ is either disjoint from $F$, that is $F \wedge R\equiv \bot$, or contained in $F$, that is $R \models F$.\footnote{The 
product of $G(Y)$ and $G'(Y')$ is a Boolean function over $Y \cup Y'$ 
whose models are exactly those assignments of $Y \cup Y'$ 
whose restriction to $Y$ ($Y'$) models $G$ ($G'$).}  
And second (Lemma~\ref{lemma:drc}), the pairs of factors 
$G(Y)$ and $G'(Y')$  of $F(Y,Y')$ satisfying the latter 
condition, call them implicants, form a disjoint rectangle cover of $F$, that is, 
$$F \equiv \bigvee_{\text{$(G,G')$ implicant}}( G(Y) \wedge G'(Y') )\text{,}$$
where the disjunction is deterministic and the conjunctions are decomposable (structured, indeed).

We then elaborate on the main technical lemma of Jha and Suciu \cite[Lemma~2.12]{DBLP:conf/icdt/JhaS12} 
to turn the above structural insight into a compilation of small size.  
We show that a circuit of small treewidth computing a function $F$ 
naturally delivers a variable tree where the number of cofactors of $F$ 
generated by assigning the variables below every node in the tree is small (Lemma~\ref{lemma:jhasuciu}).  
It follows that the disjoint rectangle covers described above are small for every factor of $F$; 
we then obtain the desired compilation by an inductive construction up the variable tree (Lemma~\ref{lemma:dsdnnf-sub} 
and Theorem~\ref{th:factorized-forms-main}).

Indeed Result~1 is proved for a more basic canonical 
deterministic structured class of circuits (Theorem~\ref{th:factorized-forms-main} 
and Proposition~\ref{prop:fiw-ctw}), which is equivalent to SDDs as far as the boundedness of their widths, 
and reduces to OBDDs in the special case of linear variable trees; 
the class is of independent interest and gives a fresh structural 
insight into SDDs (see also the conclusion).

Our construction, 
significantly shorter to describe and easier to analyze than its precursors, 
effectively encompasses the construction by Jha and Suciu in that, if carried out in the special case of circuit pathwidth, 
it compiles a circuit of $n$ variables and pathwidth $k$ 
into an OBDD (not just an SDD) of width $f(k)$ and size $O(f(k)n)$.

\subsection*{Organization}

The article is organized as follows.  The required 
notions from knowledge compilation and communication complexity 
are given in Section~\ref{sect:prel}. The part of the article devoted 
to circuit treewidth and sentential decision (Section~\ref{sect:part1}) 
deals first with the introduction and the development of the notion 
of factor and factor width for a Boolean function, 
and the relation of the latter with circuit treewidth (Section~\ref{sect:fwbf}); 
next, it presents the actual compilations in canonical 
deterministic structured NNFs and SDDs (Section~\ref{sect:impl-sdnnf}).
Section~\ref{sect:part2} is devoted to query compilation.  We present 
questions and directions for future research in Section~\ref{sect:concl}.

\section{Preliminaries}\label{sect:prel}

For every integer $n \geq 1$, we let $[n]=\{1,\ldots,n\}$.  
We refer the reader to a standard source for the 
notions of treewidth, tree decomposition, and nice tree decomposition \cite{Kloks}.

\subsection{Circuits, Determinism, Structuredness}

We consider (Boolean) circuits over the standard basis, 
namely DAGs whose non-source nodes, called internal gates, 
are unbounded fanin conjunction ($\wedge$) 
and disjunction ($\vee$) gates and fanin $1$ negation ($\neg$) gates, 
and whose source nodes, called input gates, are pairwise distinct variables 
or constants ($\bot$ and $\top$).  A designated sink node is called the output gate.  
A circuit is in \emph{negation normal form}, in short an \emph{NNF}, 
if its negation gates are only wired by input gates.  
The size $|C|$ of a circuit $C$ is the number of its gates. 

Let $X$ be a finite set of variables.  A circuit $C$ over $X$ 
is a circuit whose input gates are labelled by variables in~$X$ or by constants.    
A circuit $C$ on $X$ computes a Boolean function $F_C=F_C(X)$ over the Boolean variables $X$, 
$$F_C \colon \{0,1\}^X \to \{0,1\}\text{,}$$ 
in the usual way. We let $$\mathsf{sat}(C)=\mathsf{sat}(F_C)=F_C^{-1}(1) \subseteq \{0,1\}^X$$ denote 
the models of $C$ and $F_C$.  Two circuits $C$ and $C'$ over $X$ are equivalent, 
in symbols $C \equiv C'$, if $\mathsf{sat}(C)=\mathsf{sat}(C')$; we also write $C \equiv F$ or say that $C$ computes $F$ if $F_C=F$.  

For a gate $g$ in a circuit $C$ over $X$, 
we let $C_g$ denote the subcircuit of $C$ rooted at $g$.  In particular, $C_g=C$ if $g$ is the output gate of $C$.  
For a circuit $C$ over variables $X$ and a gate $g \in C$, 
we let $\mathsf{var}(C_g) \subseteq X$ denote the variables appearing at input gates of $C_g$.

Let $g$ be an $\vee$-gate in a circuit $C$, 
and let $h$ and $h'$ be two distinct gates wiring $g$ in $C$.  Then $g$ is called \emph{deterministic} if 
$\mathsf{sat}(C_{h}) \cap \mathsf{sat}(C_{h'})=\emptyset$, 
viewing each circuit involved in the equation as a circuit over $\mathsf{var}(C)$. 
The determinism of $g$ in $C$ implies that the two subcircuits $C_{h}$ and $C_{h'}$ 
are \lq\lq independent\rq\rq\ in the sense that 
$$|\mathsf{sat}(C_{h} \vee C_{h'})|=|\mathsf{sat}(C_{h})|+|\mathsf{sat}(C_{h'})|\text{,}$$
where each circuit involved in the equation is viewed as a circuit over $\mathsf{var}(C)$. 
A circuit where all $\vee$-gates are deterministic is called \emph{deterministic}. 

Let $Y$ be a finite nonempty set of variables.  
A \emph{variable tree} (in short, a \emph{vtree}) for the variable set $Y$ is a 
rooted, 
ordered, binary tree $T$ whose leaves correspond bijectively to $Y$; for simplicity, 
we identify each leaf in $T$ with the variable in $Y$ it corresponds to.  For technical convenience, 
we slightly relax the standard definition not requiring for a vtree to be a full binary tree. 

For every internal node $v$ of the vtree $T$ with two children, we let
$v_l$ and $v_r$ denote resp.\ the left and right child of $v$. 
Moreover, we denote by $T_{v}$ the subtree of $T$ rooted at node $v$, 
and by $Y_v \subseteq Y$ (the variables corresponding to) the leaves of $T_v$.

Let $C$ be a circuit over the variable set $X$, 
and let $T$ be a vtree for the variable set $Y$.  
Let $g$ be a fanin $2$ $\wedge$-gate in $C$, 
having wires from gates $h$ and $h'$, 
and let $v \in T$ have two children $v_l$ and $v_r$. 
We say that $g$ is \emph{structured by $v$} if 
$\mathsf{var}(C_{h}) \subseteq Y_{v_l}$ 
and $\mathsf{var}(C_{h'}) \subseteq Y_{v_r}$.  
We say that $C$ \emph{structured by $T$} if each $\wedge$-gate in $C$ 
(has fanin $2$ and) is structured by some node in $T$.  
A circuit is called \emph{structured} if it is structured by some vtree.  

A class of structured NNFs is \emph{canonical} if, 
for every Boolean function $F(X)$ and vtree $T(Y)$ with $X \subseteq Y$, 
if two circuits $C$ and $C'$ in the class both compute $F$ and are structured by $T$, 
then they are syntactically equal (not just semantically equivalent).

Note that if a gate $g \in C$ is structured, then it is also \emph{decomposable}, 
i.e., every two distinct gates $h$ and $h'$ wiring $g$ satisfy $\mathsf{var}(C_{h}) \cap \mathsf{var}(C_{h'})=\emptyset$. 
The decomposability of $g$ in $C$ implies that the two subcircuits $C_{h}$ and $C_{h'}$ 
are \lq\lq independent\rq\rq\ in the sense that 
$$|\mathsf{sat}(C_{h} \wedge C_{h'})|=|\mathsf{sat}(C_{h})| \cdot |\mathsf{sat}(C_{h'})|\text{,}$$
when viewing each circuit involved in the equation as a circuit over its own variables.   

A \emph{sentential decision diagram}, in short \emph{SDD}, 
is a deterministic structured NNF $C$ over $X$ of the form 
\begin{equation}\label{eq:sdd-def}
\bigvee_{i \in [m]}( P_i \wedge S_i )\text{,}
\end{equation}
structured by a vtree $T(Y)$, $X \subseteq Y$, such that the following holds.  
There exists a node $v \in T$ with two children $w$ and $w'$ 
structuring each $\wedge$-gate appearing in (\ref{eq:sdd-def}), 
the $P_i$'s are SDDs over $Y_{w}$ structured by $T_w$, 
the $S_i$'s are SDDs over $Y_{w'}$ structured by $T_{w'}$,
and moreover:
\begin{itemize}
\item[(1)] $\bigvee_{i \in [m]}P_i \equiv \top$;
\item[(2)] $P_i \wedge P_j \equiv \bot$ for all $i \neq j$ in $[m]$.
\end{itemize}
Constants ($\bot$ and $\top$) are SDDs (over any variable set) structured by any vtree, 
and a literal ($x$ or $\neg x$) is an SDD (over any variable set containing $x$) structured by any vtree containing $x$.  
SDDs become canonical forms if, in addition to the above, the following holds \cite{DBLP:conf/ijcai/Darwiche11}:
\begin{itemize}
\item[(3)] $S_i \not\equiv S_j$ for all $i \neq j$ in $[m]$.
\end{itemize}

\subsection{Rectangles, Covers, Complexity}

Let $X$ be a finite set of variables.  
A partition of $X$ is a sequence of pairwise disjoint subsets (blocks) of $X$ whose union is $X$.  
Let $(X_1,X_2)$ be a partition of $X$.  For $b_1 \colon X_1 \to \{0,1\}$ 
and $b_2 \colon X_2 \to \{0,1\}$, we let $b_1 \cup b_2 \colon X_1 \cup X_2 \to \{0,1\}$ 
denote the assignment of $X$ whose restriction to $X_i$ equals $b_i$ for $i=1,2$.  
Also, for $B_1 \subseteq \{0,1\}^{X_1}$ and $B_2 \subseteq \{0,1\}^{X_2}$, 
we let $B_1 \times B_2=\{ b_1 \cup b_2 \colon b_1 \in B_1, b_2 \in B_2 \}$. 
A \emph{(combinatorial) rectangle} over $X$ 
is a Boolean function $R=R(X)\colon \{0,1\}^X \to \{0,1\}$ over the Boolean variables $X$   
such that there exist a partition $(X_1,X_2)$ of $X$  
and Boolean functions $R_i \colon \{0,1\}^{X_i} \to \{0,1\}$ for $i=1,2$ such that 
$\mathsf{sat}(R) = \mathsf{sat}(R_1) \times \mathsf{sat}(R_2)$.  
We also call a subset $S$ of $\{0,1\}^X$ a rectangle over $X$, 
with underlying partition $(X_1,X_2)$, if there exists a rectangle $R \colon \{0,1\}^X \to \{0,1\}$, 
with underlying partition $(X_1,X_2)$, such that $S=\mathsf{sat}(R)$.

Let $F=F(X)$ be a Boolean function over the Boolean variables $X$.  
A finite set $\{R_i \colon i \in [m]\}$ of rectangles over $X$ 
is called a \emph{rectangle cover} of $F$ if
\begin{equation}\label{eq:rect-cover}
\mathsf{sat}(F) = \bigcup_{i \in [m]} \mathsf{sat}(R_i) \text{;} 
\end{equation}
the rectangle cover is called \emph{disjoint} 
if the union in (\ref{eq:rect-cover}) is disjoint. 
Disjoint rectangle covers and deterministic structured NNFs 
are tightly related. 

\begin{theorem}\cite{DBLP:conf/aaai/PipatsrisawatD10,BovaEtAlIJCA16}\label{th:struct}
Let $C$ be a (deterministic) structured NNF computing a function $F=F(X)$ and respecting a vtree $T$ for $X$.  
For every node $v \in T$, 
$F$ has a (disjoint) rectangle cover of size at most $|C|$ where each rectangle has underlying partition $(X_v,X\setminus X_v)$.
\end{theorem}

Let $F(X)$ be a Boolean function, 
and let $(X_1,X_2)$ be a partition of $X$ where $|X_1|=|X_2|=n$.  
The \emph{communication matrix} 
of $F$ relative to $(X_1,X_2)$, denoted by $\mathsf{cm}(F,X_1,X_2)$  is a 
Boolean matrix whose rows and columns are indexed by 
Boolean assignments of $X_1$ and $X_2$, resp., 
and whose $(b_1,b_2)$th entry equals $F(b_1 \cup b_2)$.  
We regard communication matrices as matrices over the reals.

A basic fact in communication complexity is that the rank of the communication matrix is a lower bound on the size 
of disjoint rectangle covers of a function. 

\begin{theorem}\cite[Section~4.1]{Jukna}\label{th:rklb}
Let $(X_1, X_2)$ be a partition of the variables of a function $F$, where $|X_1|=|X_2|=n$. 
Every disjoint rectangle cover of $F$ into rectangles with underlying partition $(X_1, X_2)$ 
contains at least $\mathsf{rank}(\mathsf{cm}(F, X_1, X_2))$ rectangles.
\end{theorem}

A typical application of the above statement 
is the \emph{disjointness function},  
\begin{equation}\label{eq:df}
D_n(X_n,Y_n)=(\neg x_1\lor \neg y_1) \land \cdots \land (\neg x_n\lor \neg y_n)\text{,} 
\end{equation}
where $X_n=\{x_1,\ldots,x_n\}$ and $Y_n=\{y_1,\ldots,y_n\}$.  
It is folklore that 
\begin{equation}\label{eq:rkdf}
\mathsf{rank}(\mathsf{cm}(D_n,X_n,Y_n))=2^n\text{,} 
\end{equation}
i.e., the communication matrix of $D_n$ relative to $(X_n,Y_n)$ has full rank \cite[Exercise~7.1]{Jukna}.  
Thus every disjoint rectangle cover of $D_n$ into rectangles with underlying partition $(X_n, Y_n)$ 
has at least $2^n$ rectangles.

\section{Circuit Treewidth}\label{sect:part1}

In this section, we introduce the notion of factor width of a Boolean function, 
relate it with its circuit treewidth, and show that, when parameterized by its factor width, a Boolean function 
admits a linear 
size compilation into some natural classes of canonical deterministic structured NNFs (including SDDs).  

\subsection{Factor Width and Circuit Treewidth}\label{sect:fwbf}

We introduce the notion of \emph{factor width}, 
and recall from the literature the notion of circuit treewidth \cite{DBLP:conf/icdt/JhaS12}.

Let $F(X)=F \colon \{0,1\}^X \to \{0,1\}$ be a Boolean function over a finite set of variables $X$.  For a set of variables $Y$ 
we use the notation $$F(X)=F(Y \cap X,X\setminus Y)\text{}$$
to display a partition of the variables of $F$ into the two blocks $Y \cap X$ and $X\setminus Y$.  It is intended that 
if $b \colon Y \cap X \to \{0,1\}$ and $b' \colon X\setminus Y \to \{0,1\}$, 
then $F(b,b')=F(b \cup b')$.

The \emph{cofactor} (or \emph{subfunction}) of $F$ induced by $b \colon Y \cap X \to \{0,1\}$ is the Boolean function $F'=F'(X\setminus Y)=F' \colon \{0,1\}^{X\setminus Y} \to \{0,1\}$ 
such that
$$F'(b')=F(b,b')\text{,}$$
for all $b' \colon X\setminus Y \to \{0,1\}$.  A function $F''(X\setminus Y)$ is called a \emph{cofactor of $F(X)$ relative to $X\setminus Y$}
if it is equal to the cofactor of $F$ induced by some $b \colon Y \cap X \to \{0,1\}$.

\begin{example}\label{ex:cofactors}
Let $F(x,y)=x \to y$ be Boolean implication, 
i.e.\ $F(b,b')=1$ iff $b \leq b'$, 
for all $b,b'\in \{0,1\}$.  
The cofactors of $F$ relative to $y$, 
induced by Boolean assignments of $x$, 
are $F(0,y) \equiv \top(y)$ and $F(1,y) \equiv y$.  
The cofactors of $F$ relative to $x$, 
induced by Boolean assignments of $y$, 
are $F(x,0) \equiv \neg x$ and $F(x,1) \equiv \top(x)$.  
The cofactors of $F$ induced by Boolean 
assignments of both $x$ and $y$ are $F(1,0) \equiv \bot$ 
and $F(0,0)=F(0,1)=F(1,1) \equiv \top$.  The only cofactor of $F$ 
induced by Boolean assignments of no variables is $F(x,y)$ itself.
\end{example}

We introduce notation to denote the cofactor 
of a Boolean function $F(X)$ generated by replacing, in a partition of $X$, 
the variables in certain blocks by constants. 
Let $\{Y_1,\ldots,Y_l\}$ be a partition of $X$, 
let $L \subseteq [l]$, and let $b_i \colon Y_i \to \{0,1\}$ 
for all $i \in L$.  We write 
$F(B_1,\ldots,B_l)$ 
where $B_i=b_i$ if $i \in L$ and $B_i=Y_i$ otherwise, for all $i \in [l]$, 
to denote the cofactor  of $F$ induced by $\bigcup_{i \in L} b_i$,  
i.e.\ the Boolean function over the variables $\bigcup_{i \in [l]\setminus L} Y_i$ defined, 
for every $\bigcup_{i \in [l]\setminus L} (b_i \colon Y_i \to \{0,1\})$, by
$$\bigcup_{i \in [l]\setminus L} b_i \mapsto F(b_1,\ldots,b_l)\text{.}$$

\begin{example}\label{ex:notation}
Let $F$ be as in Example~\ref{ex:cofactors}, so that $F=F(X)$ for $X=\{x,y\}$.  
Let $Y=\{x\}$.  We write $F=F(X)=F(Y,X\setminus Y)$.  
Let $b \colon Y \to \{0,1\}$.  We write $F(b,X\setminus Y)$ 
to denote the cofactor of $F$ relative to $X\setminus Y$ induced by $b$.  If $b(x)=0$, 
then $F(b,X\setminus Y) \equiv \top(y)$.
\end{example}

Let $Y \subseteq X$.  Intuitively, a factor of 
$F(Y,X\setminus Y)$ is a function $G(Y)$ 
whose models correspond exactly to the assignments of $Y$ 
that induce some fixed cofactor $F'(X\setminus Y)$ of $F$.

\begin{definition}\label{def:factor}
Let $Y$ and $X$ be finite sets of variables 
and let $F(X)=F(Y \cap X,X\setminus Y)$ be a Boolean function.
A Boolean function $G=G(Y \cap X)=G\colon \{0,1\}^{Y \cap X} \to \{0,1\}$ is called a \emph{factor of $F(X)$ relative to $Y$} if 
there exists a cofactor $F'=F'(X \setminus Y)$ of $F$ such that 
$$b \in \mathsf{sat}(G) \Longleftrightarrow F(b,X \setminus Y)=F'\text{.}$$
The factors of $F$ relative to $Y$ are denoted by $\mathsf{factors}(F,Y)$. 
\end{definition}

\begin{example}\label{ex:factors}
Let $F(x,y)=x \to y$ be Boolean implication as in Example~\ref{ex:cofactors}.  
The function $G(x) \equiv x$ is a factor of $F$ relative to $x$, 
because there exists a cofactor of $F$ relative to $y$, 
namely $F'(y) \equiv y$, such that $b \models G(x)$ 
iff $F(b,y)=F'(y)$.  
The function $G(x) \equiv \neg x$ is a factor of $F$ relative to $x$, 
because there exists a cofactor of $F$ relative to $y$, 
namely $F'(y) \equiv \top(y)$, such that $b \models G(x)$ 
iff $F(b,y)=F'(y)$.  
\end{example}

Note that cofactors and factors of a Boolean function relative to a variable set are, in general, 
distinct.\footnote{Exceptions include the parity function.}  

\begin{example}\label{ex:fneqcof}
Let $F(x,y)=x \to y$ be Boolean implication as in Example~\ref{ex:cofactors}.  
Then $G(x) \equiv x$ is a factor of $F$ relative to $x$ (Example~\ref{ex:factors}), 
but it is not a cofactor of $F$ relative to $x$, 
since the only cofactors of $F$ relative to $x$ 
are equivalent to $\neg x$ and $\top(x)$ (Example~\ref{ex:cofactors}).
 \end{example}

Note that, by Definition~\ref{def:factor}, 
\begin{equation}\label{eq:dummyvars}
\mathsf{factors}(F,Y)=\mathsf{factors}(F,Y\cap X)\text{,} 
\end{equation}
but we insist on $Y$ being an arbitrary set of variables for technical convenience.  
Moreover, again by Definition~\ref{def:factor}, 
\begin{equation}\label{eq:part}
\{0,1\}^{Y \cap X} = \bigcup_{G \in \mathsf{factors}(F,Y)} \mathsf{sat}(G) 
\end{equation}
and the union is disjoint.  In words, $\mathsf{factors}(F,Y)$ naturally determines 
a partition of $\{0,1\}^{Y \cap X}$ whose blocks, of the form 
$\mathsf{sat}(G)$ for $G \in \mathsf{factors}(F,Y)$, 
correspond to the cofactors of $F$ relative to $X \setminus Y$.

Finally, we introduce the notion of factor width of a Boolean function.  

\begin{definition}
Let $F=F(X)$ be a Boolean function and let $T$ be a vtree for $Z \supseteq X$.  
The \emph{factor width of $F$ relative to $T$}, in symbols $\fw(F,T)$, 
is defined by  
$$\fw(F,T)= \max_{v \in T}|\mathsf{factors}(F,Z_v)|\text{.}$$
The \emph{factor width} of $F$ is defined by 
$$\fw(F)= \min\{ \fw(F,T) \colon \textup{$T$ vtree for $Z \supseteq X$} \}\text{.}$$ 
\end{definition}

The \emph{treewidth of a circuit} $C$, 
in symbols $\tw(C)$, is the treewidth  
of 
the undirected graph underlying (the directed acyclic graph underlying) $C$.  
The \emph{circuit treewidth} $\ctw(F)$ of a Boolean function $F$ 
is the minimum treewidth of a circuit computing $F$.

A crucial fact in our development is that the factor width of a Boolean function is bounded above by a function 
of its circuit treewidth. The proof is a revisitation of \cite[Lemma~2.12]{DBLP:conf/icdt/JhaS12}.

\begin{lemma}\label{lemma:jhasuciu}
For all Boolean funct\-ions $F$,  
$$\fw(F) \leq 2^{(\ctw(F)+2)2^{\ctw(F)+1}} \text{.}$$
\end{lemma}
\begin{proof} 
Let $C$ be a treewidth $k-1$ circuit computing the Boolean function $F(X)$.  
Let $S$ be 
a \emph{nice} tree decomposition of the gates of $C$, 
witnessing treewidth $k-1$; without loss of generality, the root of $S$ is the empty bag, 
therefore each input gate of $C$ (i.e., each variable in $X$) is forgotten exactly once in $S$.  

We associate to $S$ a vtree $T$ for $X$ as follows.  Let $W$ be a set of fresh variables 
in a bijective correspondence with the leaves of $S$. 
Label the leaves of $S$ by pairwise distinct (dummy) variables in $W$.  For every variable 
$x \in X$, append a fresh leaf labelled $x$ to the node forgetting $x$ in $S$. 
The resulting tree $T$ is a vtree for $X \cup W \supseteq X$.  

For every $v \in T$, let $X_v$ denote the variables in $X$ appearing in $T_v$ 
(or equivalently, the variables in $T_v$ that are not dummy variables).  
In light of (\ref{eq:dummyvars}), to show that $$\fw(F,T) \leq 2^{(k+1)2^{k}}\text{}$$
it is sufficient to prove that $|\mathsf{factors}(F,X_v)|$ matches the bound for all $v \in T$.  

If $v$ is a leaf in $T$, then $X_v=\{x\}$ for some $x \in X$ 
and $|\mathsf{factors}(F,X_v)| \leq 2$, 
or $v$ is labelled by some dummy variable in $W$, 
so that $X_v=\emptyset$ and $|\mathsf{factors}(F,X_v)|=1$.  
Otherwise $v$ is a bag $B$ 
in $S$ and in this case $X_v$ contains the variables in $X$ 
forgotten by nodes in the subtree of $S$ rooted at $v$.  

For a gate $g$ in $C$, let $K(g) \subseteq B$ be the gates 
in $B$ with a directed path to $g$ in (the DAG underlying) $C$ whose intermediate gates are not in $B$.
Namely, $h \in K(g)$ iff $h \in B$ and there exists a directed path 
$$h \stackrel{C}{\to} h_1 \stackrel{C}{\to} \cdots \stackrel{C}{\to} h_m \stackrel{C}{\to} g$$
in the DAG underlying $C$ such that $\{h_1,\ldots,h_m\} \cap B=\emptyset$.

Let $g$ be a gate in $C$.  We freely identify 
$C_g$, the subcircuit of $C$ rooted at $g$, 
with the Boolean function on $X$ it computes.  

\begin{claim}\label{cl:main-cl}
For every $g \in C$, it holds that 
\begin{equation}\label{eq:rec}
|\mathsf{factors}(C_g,X_v)| \leq 2^{2^{|K(g)|}} \prod_{h \in K(g)}|\mathsf{factors}(C_h,X_v)| \text{.}
\end{equation} 
\end{claim}
\begin{proof}[Proof of Claim~\ref{cl:main-cl}]
Let $Z=\{z_h \colon h\in K(g)\}$ be a set of fresh variables.  
Let $C_g'$ denote the circuit obtained from $C_g$ 
by transforming each gate $h \in K(g)$ into an input gate labelled by the variable $z_h \in Z$.  
We distinguish two cases depending on whether or not $g$ is in $B$.

If $g \not\in B$, then observe that $\var(C'_g) \subseteq X_v \cup Z$, 
or $\var(C'_g) \subseteq (X\setminus X_v ) \cup Z$; otherwise, if $C'_g$ 
uses variables in both $X_v$ and $X\setminus X_v$ then, in the graph underlying $C$, 
there exists a path from $X_v$ to $X\setminus X_v$ not intersecting $B$, 
contradicting the properties of 
$S$.

If $\var(C'_g) \subseteq X_v \cup Z$ then, for every $b \colon X_v \to \{0,1\}$, 
$C'_g(b)$ is a function of $Z$ among at most $$2^{2^{|Z|}}=2^{2^{|K(g)|}}$$ possibilities.  
Otherwise, if $\var(C'_g) \subseteq (X\setminus X_v ) \cup Z$, 
then every $b \colon X_v \to \{0,1\}$ yields the same function 
$C'_g(b)=C'_g$ of $(X\setminus X_v) \cup Z$.  Since each variable $z_h \in Z$ 
in $C'_g$ represents the subcircuit $C_h$ of $C$ for $h \in K(g)$, it follows that  
$$|\mathsf{factors}(C_g,X_v)| \leq 2^{2^{|K(g)|}} \prod_{h \in K(g)}|\mathsf{factors}(C_h,X_v)| \text{.}$$

If $g \in B$, then each proper subcircuit of $C'_g$ 
uses only variables in $X_v \cup Z$ or only variables in $(X\setminus X_v) \cup Z$; 
otherwise we obtain a contradiction as above.  Let $G_1,\ldots,G_m$ 
and $H_1,\ldots,H_l$ be a bipartition of the immediate subcircuits of $C'_g$ 
such that the $G_i$'s only use variables in $X_v \cup Z$ 
and the $H_i$'s only use variables in $(X\setminus X_v) \cup Z$.  

If $g$ is an $\wedge$-gate, then $$C'_g \equiv (G_1 \wedge \cdots \wedge G_m) \wedge (H_1 \wedge \cdots \wedge H_l)\text{.}$$
For every $b \colon X_v \to \{0,1\}$, 
$(G_1 \wedge \cdots \wedge G_m)(b)$ is a function $G_b$ 
over $Z$; 
on the other hand, every $b \colon X_v \to \{0,1\}$ yields the same function 
$H$ over $(X\setminus X_v) \cup Z$, namely $H_1 \wedge \cdots \wedge H_l$.  Therefore, 
for all $b \colon X_v \to \{0,1\}$, $C'_g(b)$ is a function over $(X\setminus X_v) \cup Z$ of the form 
$$G_b \wedge H\text{,}$$ 
where $G_b$ is a function over $Z$, so that 
$$|\{C'_g(b) \colon b \in \{0,1\}^{X_v} \}| \leq 2^{2^{|Z|}}\text{.}$$
Recalling that each variable $z_h \in Z$ in $C'_g$ 
represents the subcircuit $C_h$ of $C$ for $h \in K(g)$, we have
$$|\mathsf{factors}(C_g,X_v)| \leq 2^{2^{|K(g)|}} \prod_{h \in K(g)}|\mathsf{factors}(C_h,X_v)| \text{.}$$

The cases where $g$ is a $\vee$-gate or a $\neg$-gate are similar, and the claim is proved.
\end{proof}

For all $g \in C$, let $L(g)=(C_g\setminus \{g\}) \cap B$.  By induction on $|L(g)| \geq 0$ we prove that for all $g \in B$ 
it holds that 
\begin{equation}\label{eq:ind}
|\mathsf{factors}(C_g,X_v)| \leq 2^{2^{|L(g)|}}\text{.} 
\end{equation}

For the base case, let $g \in B$ be such that $|L(g)|=0$.  Then 
$$|\mathsf{factors}(C_g,X_v)| \leq 2^{2^0} \cdot 1=2^{2^{|L(g)|}}\text{,}$$
where the first inequality holds by (\ref{eq:rec}) as $K(g)=\emptyset$ in this case.

For the inductive case, let $g \in B$ be such that $|L(g)| \geq 1$.  Then, resp.\ by (\ref{eq:rec}) 
and the induction hypothesis as $|L(h)|<|L(g)|$ for all $h \in K(g)$,
\begin{align*}
|\mathsf{factors}(C_g,X_v)| &\leq 2^{2^{|K(g)|}} \prod_{h \in K(g)}|\mathsf{factors}(C_h,X_v)| \\
&\leq 2^{2^{|K(g)|}} \prod_{h \in K(g)} 2^{2^{|L(h)|}}\text{.}
\end{align*}
But  
$$2^{2^{|K(g)|}} \prod_{h \in K(g)} 2^{2^{|L(h)|}} \leq 2^{2^{|K(g)|+ \sum_{h \in K(g)} |L(h)|}} \leq 2^{2^{|L(g)|}}$$
where $|K(g)|+ \sum_{h \in K(g)} |L(h)| \leq |L(g)|$ justifies the last inequality.

We now conclude the proof.  Let $g$ be the output gate of $C$, i.e., $C=C_g$.  Then, 
justifying the first and second inequalities resp.\ by (\ref{eq:rec}) and by (\ref{eq:ind}), 
$K(g) \subseteq B$ and $|B|\leq k$, 
\begin{align*}
|\mathsf{factors}(F,X_v)| & = |\mathsf{factors}(C_g,X_v)|\\
& \leq 2^{2^{|K(g)|}} \prod_{h \in K(g)}|\mathsf{factors}(C_h,X_v)|\\
&\leq 2^{2^k} \left( 2^{2^{k}} \right)^k\\
&= 2^{(k+1)2^k}\text{,} & 
\end{align*}
and we are done.
\end{proof}

We conclude the section observing that circuit tree\-wi\-dth is computable, 
thus answering a question posed by Jha and Suciu \cite[Section~6]{DBLP:conf/icdt/JhaS12}. 

\begin{proposition}\label{prop:seese}
The circuit treewidth of a Boolean function is computable.
\end{proposition}
\begin{proof}
Let $X=\{x_1,\ldots,x_n\}$ be a variable set.  Say that a graph (with loops) $G$ \emph{implements} a circuit over $X$ 
iff  there exists a circuit $C$ over $X$ such that $G$ is (isomorphic to) 
the graph whose edge set $E$ is defined relative to $C$ as follows: 
\begin{itemize}
\item For every wire $g \to g'$ in $C$, add to $E$ 
a path $\{g,h\}$, $\{h,h'\}$, $\{h',g'\}$ from $g$ to $g'$ with a loop on $h'$, 
where $h$ and $h'$ are fresh vertices.
\item Add to $E$ a loop on the output gate $g$ of $C$.
\item Let 
\begin{align*}
(\circ,j) &\in \{(x_1,1),\ldots,(x_n,n),(\bot,n+1),(\top,n+2)\}\\ 
           &\ \ \ \cup \{ (\neg,n+3),(\wedge,n+4),(\vee,n+5) \}\text{.} 
\end{align*} 
For every $\circ$-gate $g$ in $C$, add to $E$ a $j$-star $\{g,h_1\}$, $\ldots$, $\{g,h_j\}$ centered at $g$, 
where $h_1,\ldots,h_j$ are fresh vertices.
\end{itemize}

Intuitively a graph (with loops) $G$ implementing a circuit $C$ 
is a faithful representation of $C$ in the vocabulary of graphs, 
where the arcs and labels used to represent $C$ in the vocabulary of circuits are suitably expressed by edges (and loops).   
Let $F$ be a Boolean function over $X$.  It is a tedious but straightforward exercise to write an MSO sentence $\phi_F$ 
that is true on a graph (with loops) $G$ iff  $G$ implements a circuit $C$ on $X$ computing $F(X)$ \cite[Examples 4.10, 4.13, and 4.18]{FlumGroheBook}.\footnote{Recall that a 
(finite) digraph is acyclic iff  every induced subgraph has a source and a sink.}

Also note that, if the graph (with loops) $G$ implements a circuit $C$, 
then the treewidth of $G$ is equal to the treewidth of 
$C$.  
Now, let $k$ be an upper bound on the circuit treewidth of $F$ 
(for instance, the treewidth of the DNF whose terms are exactly the models of $F$).  
Seese proves that, given an MSO sentence on the vocabulary of graphs, 
it is decidable whether it is satisfied by a graph (with loops) of treewidth $k$ \cite{DBLP:journals/apal/Seese91}.  
We therefore cycle for $i=1,2,\ldots$ until we find $i \leq k$ such that $\phi_F$ 
is modeled by a graph of treewidth $i$.
\end{proof}

\subsection{Few Factors Imply Small Disjoint Rectangle Covers}\label{sect:impl-sdnnf}

We show that Boolean functions of small factor width 
have implementations of small width within natural canonical subclasses of 
deterministic structured NNFs (including canonical SDDs); conversely, small width implementations 
within such circuit classes imply small circuit treewidth.  Thus a class of Boolean functions 
has bounded circuit treewidth iff  it has bounded width implementations in some natural, canonical classes of 
deterministic structured NNFs, including SDDs.

\subsubsection{Factorized Implicant Width and Deterministic Structured Forms}\label{sect:fimpldec}

We introduce the notion of \emph{factorized implicant} of a Boolean function $F(Y,Y')$, 
roughly
$$F'(Y) \wedge F''(Y') \models F(Y,Y')\text{,}$$ 
where $Y$ and $Y'$ are disjoint sets of variables, 
$F'$ (resp., $F''$) is a factor of $F$ relative to $Y$ (resp., $Y'$).  

The first key insight is that the rectangle formed by multiplying any two factors $F'(Y)$ and $F''(Y')$ of $F(Y,Y')$ 
is either disjoint from $F$ or contained in $F$.  

\begin{lemma}\label{lemma:factors}
Let $F=F(X)$ be a Boolean function.  Let $Y$ and $Y'$ be disjoint subsets of $X$.  
Let $H$, $G$, and $G'$ be factors of $F$ relative to $Y \cup Y'$, $Y$, and $Y'$, respectively.  Then either 
\begin{equation}\label{eq:rect-cont}
(\mathsf{sat}(G) \times \mathsf{sat}(G')) \subseteq  \mathsf{sat}(H)
\end{equation}
or 
\begin{equation}\label{eq:rect-disj}
(\mathsf{sat}(G) \times \mathsf{sat}(G')) \cap  \mathsf{sat}(H) = \emptyset\text{.}
\end{equation}
\end{lemma}
\begin{proof}
If $(\mathsf{sat}(G) \times \mathsf{sat}(G')) \cap \mathsf{sat}(H) \neq \emptyset$, 
then let $b \in \mathsf{sat}(G)$ and $b' \in \mathsf{sat}(G')$ be such that 
$b \cup b' \in \mathsf{sat}(H)$.  Let $c \in \mathsf{sat}(G)$ and $c' \in \mathsf{sat}(G')$.  
It suffices to show that $c \cup c' \in \mathsf{sat}(H)$.  Below, $X'=X\setminus (Y \cup Y')$.

Assume for a contradiction that $c \cup c' \not\in \mathsf{sat}(H)$.  
Therefore, by definition, the cofactors of $F$ induced by $b \cup b'$ and $c \cup c'$ are distinct, 
i.e., $$F(b,b',X') \neq F(c,c',X')\text{.}$$  

On the other hand, by definition, $\{b,c\} \subseteq \mathsf{sat}(G)$ implies that $F(b,Y',X')=F(c,Y',X')$, 
and similarly $\{b',c'\} \subseteq \mathsf{sat}(G')$ implies that $F(Y,b',X')=F(Y,c',X')$.  In particular, 
$F(b,b',X')=F(c,b',X')$ and $F(c,b',X')=F(c,c',X')$.  Then $$F(b,b',X')=F(c,c',X')\text{,}$$ a contradiction.
\end{proof}

Intuitively, a factorized implicant of a function $F(Y,Y')$ is a pair of factors $F'(Y)$ and $F''(Y')$ 
entirely contained in $F$, as in (\ref{eq:rect-cont}).  Formally,  

\begin{definition} 
Let $F=F(X)$ be a Boolean function.  
Let $H$, $G$, and $G'$ be factors of $F$ relative to $Y \cup Y'$, $Y$, and $Y'$, resp., 
where $Y$ and $Y'$ are disjoint subsets of $X$.  Then $(G,G')$ is 
a \emph{factorized implicant of $H$ relative to $(F,Y,Y')$} 
if $G$ and $G'$ satisfy (\ref{eq:rect-cont}) relative to $H$.  
We denote by $\mathsf{impl}(F,H,Y,Y')$ the set of factorized implicants of $H$ relative to $(F,Y,Y')$. 
\end{definition}

The second key insight is that 
the factorized implicants of $F(Y,Y')$ naturally induce a disjoint rectangle cover for $F(Y,Y')$.

\begin{lemma}\label{lemma:drc}
Let $F=F(X)$ be a Boolean function and let $H$ be a factor of $F$ relative to $Y \cup Y'$, 
where $Y$ and $Y'$ are disjoint subsets of $X$.  Then 
$$\{ \mathsf{sat}(G) \times \mathsf{sat}(G') \colon (G,G') \in \mathsf{impl}(F,H,Y,Y') \}$$
forms a disjoint rectangle cover of $H$, i.e., 
\begin{equation}\label{eq:main-node-1}
\mathsf{sat}(H)=\bigcup_{(G,G') \in \mathsf{impl}(F,H,Y,Y')} (\mathsf{sat}(G) \times \mathsf{sat}(G'))\text{,}
\end{equation}
and the union is disjoint.
\end{lemma}

A circuit interpretation of Lemma~\ref{lemma:drc} is that 
\begin{equation}\label{eq:syntactic-fif}
C_H \equiv \bigvee_{(G,G') \in \mathsf{impl}(F,H,Y,Y')}(C_G \wedge C_{G'})
\end{equation}
where $C_H$, $C_G$, and $C_{G'}$ are circuits using variables in $Y \cup Y'$, $Y$, and $Y'$ 
computing $H$, $G$, and $G'$, respectively. The $\wedge$-gates are decomposable as $Y$ and $Y'$ are disjoint, 
and by Lemma~\ref{lemma:drc} the $\vee$-gate is deterministic.

\begin{proof}[Proof of Lemma~\ref{lemma:drc}]
We claim that the union on the right of (\ref{eq:main-node-1}) is disjoint.  
Indeed if $(G_1,G'_1)$ and $(G_2,G'_2)$ are distinct implicants in $\mathsf{impl}(F,H,Y,Y')$, then $G_1\neq G_2$ or $G'_1\neq G'_2$.  If $G_1\neq G_2$, 
then $\mathsf{sat}(G_1) \cap \mathsf{sat}(G_2)=\emptyset$ 
because distinct factors of $F$ relative to $Y$ have disjoint models by (\ref{eq:part}).  
Similarly, if $G'_1\neq G'_2$, then $\mathsf{sat}(G'_1) \cap \mathsf{sat}(G'_2)=\emptyset$.  
Therefore, $(\mathsf{sat}(G_1) \times \mathsf{sat}(G'_1)) \cap (\mathsf{sat}(G_2) \times \mathsf{sat}(G'_2)) =\emptyset$, 
and we are done.
      
We prove the equality in (\ref{eq:main-node-1}).  For the nontrivial inclusion ($\subseteq$), 
let $b \colon Y \cup Y' \to \{0,1\}$ be in $\mathsf{sat}(H)$.  
By (\ref{eq:part}), there exist $G \in \mathsf{factors}(F,Y)$ 
such that $b|_Y \in \mathsf{sat}(G)$ and $G' \in \mathsf{factors}(F,Y')$ 
such that $b|_{Y'} \in \mathsf{sat}(G')$.  Then $b|_Y \cup b|_{Y'}=b \in (\mathsf{sat}(G) \times \mathsf{sat}(G')) \cap \mathsf{sat}(H)$.  
It follows by Lemma~\ref{lemma:factors} that $\mathsf{sat}(G) \times \mathsf{sat}(G') \subseteq \mathsf{sat}(H)$.  
Then $(G,G') \in \mathsf{impl}(F,H,Y,Y')$ and $b$ is contained in the union on the right.
\end{proof}

The above insight can be exploited recursively to implement Boolean functions 
within a natural, canonical class of deterministic structured forms.

Let $F$ be a function and let $T$ be a vtree, 
both over the variables $X$.  For every node $v \in T$ 
and every factor $H$ of $F$ relative to $X_v$, 
we construct a  
circuit $C_{v,H}$ as follows. 

If $v$ is a leaf of $T$, then $X_v=\{x\}$ for some variable $x \in X$.  
There are two cases.  Either $F(0,X\setminus\{x\})=F(1,X\setminus\{x\})$, 
or $F(0,X\setminus\{x\})\neq F(1,X\setminus\{x\})$.  In the former case, 
$\mathsf{factors}(F,\{x\})=\{H\}$ and
\begin{equation}\label{eq:base-1}
C_{v,H}=\top \text{.}
\end{equation}
In the latter case, $\mathsf{factors}(F,\{x\})=\{H_0,H_1\}$ and
\begin{align}
C_{v,H_0} &= \neg x\text{,} \label{eq:base-2}\\
C_{v,H_1} &= x\text{.}\label{eq:base-3}
\end{align}

If $v$ is a node of $T$ with children $w$ and $w'$, we put
\begin{equation}\label{eq:main-node}
C_{v,H}=\bigvee_{(G,G') \in \mathsf{impl}(F,H,X_w,X_{w'})} \left( C_{w,G} \wedge C_{w',G'} \right)\text{.}
\end{equation}

Finally we put 
\begin{equation}\label{eq:cft}
C_{F,T} = C_{r,F} \text{,}
\end{equation}
where $r$ is the root of $T$; note that $F$ itself is a factor of $F$ relative to $X$; 
its models induce the cofactor $1 \colon \{0,1\}^\emptyset \to \{0,1\}$ of $F$, 
the identically $1$ function (over $\emptyset$).

\begin{lemma}\label{lemma:dsdnnf-sub}
Let $F$ be a Boolean function and let $T$ be a vtree, both over the variables $X$.  
Let $v \in T$ and let $H \in \mathsf{factors}(F,X_v)$.  The following holds. 
\begin{itemize}
\item $C_{v,H}$ is a deterministic structured NNF respecting the vtree $T_v$.
\item $C_{v,H}$ computes $H$.
\end{itemize}
\end{lemma}
\begin{proof}
The proof is a routine induction on the depth of $v$ in $T$.  
The base case holds by inspection of (\ref{eq:base-1})-(\ref{eq:base-3}).  
The inductive case holds by inspection of (\ref{eq:main-node}), 
using (\ref{eq:main-node-1}) in Lemma~\ref{lemma:drc} and the induction hypothesis; 
indeed, note that the disjunction arising in (\ref{eq:main-node}) is deterministic because the union in (\ref{eq:main-node-1}) is disjoint, 
and the conjunctions arising in (\ref{eq:main-node}) are structured by $T_v$ 
(namely the left and right conjuncts are over variables $X_w$ and $X_{w'}$ resp., 
where $w$ and $w'$ are the children of $v$ in $T$) by the induction hypothesis.  
\end{proof}

Note that, by Lemma~\ref{lemma:dsdnnf-sub}, it follows that $C_{F,T}$ 
is a deterministic structured NNF computing $F$, canonical in that it is uniquely determined by the vtree $T$ and $F$; 
the notion of \emph{factorized implicant width} of $F$ relative to $T$ arises naturally.

\begin{definition}
Let $F$ be a Boolean function and let $T$ be a vtree, both over the variables $X$.  The \emph{factorized implicant width 
of $F$ relative to $T$}, in symbols $\fiw(F,T)$, is defined by\footnote{Recall that, if $g \in C_{F,T}$ is an $\wedge$-gate arising from (\ref{eq:main-node}), 
we say that $g$ is \emph{structured} by $v \in T$.}    
$$\fiw(F,T)= \max_{v \in T}|\{ g \in C_{F,T} \colon \textup{$g$ is structured by $v$} \}|\text{.}$$
The \emph{factorized implicant width} of $F$, in symbols $\fiw(F)$, is defined by  
$$\fiw(F)= \min\{ \fiw(F,T) \colon \textup{$T$ vtree for $X$} \}\text{.}$$ 
\end{definition}

Relative to its factorized implicant width, 
a Boolean function (of $n$ variables) has linear (in $n$) 
size compilations into canonical deterministic structured forms.

\begin{theorem}\label{th:factorized-forms-main}
A Boolean function $F$ of $n$ variables and factorized implicant width $k$ 
has canonical deterministic structured NNFs of size $O(kn)$.
\end{theorem}
\begin{proof}
Let $X$ be the variables of $F$, so that $|X|=n$, and let $T$ be a vtree for $X$ 
witnessing factorized implicant width $k$ for $F$.  The circuit $C_{F,T}$ in (\ref{eq:cft}) 
is a canonical deterministic structured NNF computing $F$ by Lemma~\ref{lemma:dsdnnf-sub}.  

Moreover, we claim that $C_{F,T}$ has size $O(kn)$.  The $n$ leaves of $T$ contribute at most $n+1$ input gates and $n$ $\neg$-gates in $C_{F,T}$.  
The $n-1$ internal nodes of $T$ contribute each at most $k$ $\wedge$-gates (by the definition of factorized implicant width), 
and each such gate is linked with at most $3$ $\vee$-gates.  Hence $C_{F,T}$ contains at most $2n+1+3k(n-1)=O(kn)$ gates.  
\end{proof}

We conclude the section showing that a class of Boolean functions has bounded circuit treewidth 
iff  it has bounded factorized implicant width.  It is sufficient to prove that 
the factorized implicant width of a Boolean function is bounded below and above 
by computable functions of its circuit treewidth.  

For the upper bound, we have 
\begin{equation}\label{eq:fiw-ub}
\fiw(F) \leq \fw(F)^2 \leq 2^{(\ctw(F)+2)2^{\ctw(F)+2}}\text{,} 
\end{equation}
where the first inequality is justified by the observation that 
every $\wedge$-gate $g$ in $C_{F,T}$, structured by a node $v \in T$, 
corresponds to a pair of factors of $F$, 
and the second inequality follows by Lemma~\ref{lemma:jhasuciu}.  

For the lower bound, we verify that 
small factorized implicant width 
implies small circuit treewidth.

\begin{proposition}\label{prop:fiw-ctw}
For all Boolean functions $F$, 
\begin{equation}\label{eq:fiw-lb}
\ctw(F)/3 \leq \fiw(F)\text{.} 
\end{equation}
\end{proposition}
\begin{proof}
Let $T$ be a vtree for $X$ such that $$\fiw(F,T)=\fiw(F)=k\text{.}$$  
We claim that $\tw(C_{F,T}) \leq 3k$, so that $\ctw(F)\leq 3k$ by Lemma~\ref{lemma:dsdnnf-sub}.

For every gate $g \in C_{F,T}$, let $\mathsf{neigh}(g)$ denote the closed neighborhood of $g$ in the undirected graph 
underlying $C_{F,T}$.  We define a tree decomposition for the undirected graph 
underlying $C_{F,T}$, as follows.  The bags of the tree decomposition have the form 
$$B_v=\{ \mathsf{neigh}(g) \colon \textup{$g$ structured by $v$} \}\text{,}$$
for all $v \in T$.  The root of the tree decomposition is $B_r$.  
The bag $B_w$ has an arc to the bag $B_{w'}$ iff  $w$ has an arc to $w'$ in $T$. 
By definition, $B_v$ contains all the $\wedge$-gates structured by the node $v \in T$; 
there are at most $k$ such gates in $C_{F,T}$ by definition, 
and each such gate has indegree $2$ and outdegree $1$ by construction (therefore degree $3$ in the undirected graph 
underlying $C_{F,T}$).  Hence $|B_v| \leq 3k$. We check the desired properties.  

Since every wire of $C_{F,T}$ enters or leaves an $\wedge$-gate and 
every $\wedge$-gate is structured by some $v \in T$, the edges of the undirected graph 
underlying $C_{F,T}$ are covered by the tree decomposition.  Moreover let $g$ be a gate of $C_{F,T}$ 
occurring in two distinct bags $B_w$ and $B_{w'}$.  Then $g \in \mathsf{neigh}(h) \cap \mathsf{neigh}(h')$ 
where $h$ is a $\wedge$-gate structured by $w$ and $h'$ is a $\wedge$-gate structured by $w'$.  
By construction of $C_{F,T}$, either $w$ has an arc to $w'$ in $T$ or $w'$ has an arc to $w$ in $T$.  
Hence $B_w$ and $B_{w'}$ are adjacent. 
\end{proof}

\subsubsection{Sentential Decision Width and Sentential Decision Diagrams}

We show that the notion of factorized implicant lies at the core 
of (and provides fresh insight on) the canonical construction of SDDs for Boolean functions.  

We prepare the actual description of the construction in two steps.  
The first step yields, by a straightforward generalization of Lemma~\ref{lemma:drc}, 
a factorized implicant decomposition reminiscent of (\ref{eq:syntactic-fif}) for unions of factors.

\begin{lemma}\label{lemma:union-drc}
Let $F=F(X)$ be a Boolean function,  
let $Y$ and $Y'$ be disjoint subsets of $X$, 
and let $\mathcal{H} \subseteq \mathsf{factors}(F,Y \cup Y')$.  Then 
$$\{ \mathsf{sat}(G) \times \mathsf{sat}(G') \colon (G,G') \in \mathsf{impl}(F,H,Y,Y'), H \in \mathcal{H} \}$$
forms a disjoint rectangle cover of $\bigvee_{H \in \mathcal{H}} H$.
\end{lemma}
In terms of circuits, the statement means that 
\begin{equation}\label{eq:syntactic-sdd}
\bigvee_{H \in \mathcal{H}} C_H \equiv \bigvee_{\substack{H \in \mathcal{H}\\(G,G') \in \mathsf{impl}(F,H,Y,Y')}} \left(C_G \wedge C_{G'}\right)\text{,}
\end{equation}
where $C_H$, $C_G$, and $C_{G'}$ are as in (\ref{eq:syntactic-fif}), the $\vee$-gate is deterministic, 
and the $\wedge$-gates are decomposable. 
\begin{proof}[Proof of Lemma~\ref{lemma:union-drc}] We have
\begin{align*}
\mathsf{sat}\left(\bigvee_{H \in \mathcal{H}} H\right) &= \bigcup_{H \in \mathcal{H}} \mathsf{sat}(H) & \text{} \\
&= \bigcup_{\substack{H \in \mathcal{H}\\(G,G') \in \mathsf{impl}(F,H,Y,Y')}} \left(\mathsf{sat}(G) \times \mathsf{sat}(G')\right) \text{,}
\end{align*}
where the second equality follows by applying Lemma~\ref{lemma:drc} to $H$.  
We claim that the union is disjoint.  Indeed, by (\ref{eq:part}), 
distinct factors of $F$ relative to $Y \cup Y'$, in particular those in $\mathcal{H}$, 
have disjoint models; 
moreover, by Lemma~\ref{lemma:drc}, distinct implicants of a factor have disjoint models.
\end{proof}

The second step enforces the properties of a proper 
sentential decision over the factorized implicant form (\ref{eq:syntactic-sdd}) 
given by Lemma~\ref{lemma:union-drc}, as follows.  For each $G \in \mathsf{factors}(F,Y)$, 
let 
\begin{align*}
\mathcal{S}_G &=\{ G' \colon (G,G') \in \mathsf{impl}(F,H,Y,Y'), H \in \mathcal{H} \} \\
&\subseteq \mathsf{factors}(F,Y') 
\end{align*}
and observe that  
\begin{equation}\label{eq:noncanonical-sdd-form}
\bigvee_{G \in \mathsf{factors}(F,Y)} \left(C_G \wedge \left(\bigvee_{G' \in \mathcal{S}_G} C_{G'}\right)\right)\text{,} 
\end{equation}
where empty disjunctions are implemented by $\bot$, is equivalent to $\bigvee_{H \in \mathcal{H}} C_H$ 
and is a sentential decision form as the factors of $F$ relative to $Y$ partition $\{0,1\}^Y$.  

However, (\ref{eq:noncanonical-sdd-form}) is not a canonical form because 
distinct $G_1$ and $G_2$ in $\mathsf{factors}(F,Y)$ 
can give $\mathcal{S}_{G_1}=\mathcal{S}_{G_2}$.  Let  
\begin{align*}
\{\mathcal{S}_1,\ldots,\mathcal{S}_m\} &= \{\mathcal{S}_G \colon G \in \mathsf{factors}(F,Y)\}\text{,}\\
\mathcal{P}_i &= \{ G \in \mathsf{factors}(F,Y) \colon \mathcal{S}_G=\mathcal{S}_i \}\text{,}
\end{align*}
for all $i \in [m]$, and 
\begin{align*}\label{eq:sent-dec}
\sd(F,\mathcal{H},Y,Y') &=\{(\mathcal{P}_1,\mathcal{S}_1),\ldots,(\mathcal{P}_m,\mathcal{S}_m)\} \\ 
&\subseteq 2^{\mathsf{factors}(F,Y)} \times 2^{\mathsf{factors}(F,Y') \cup \{\bot\}} \text{.} 
\end{align*}
It is readily observed that
\begin{equation}\label{eq:canonical-sdd-form}
\bigvee_{(\mathcal{P},\mathcal{S}) \in \sd(F,\mathcal{H},Y,Y')}   \left( \left( \bigvee_{P \in \mathcal{P}} C_P \right) \wedge \left( \bigvee_{S \in \mathcal{S}} C_S \right) \right) 
\end{equation}
is equivalent to $\bigvee_{H \in \mathcal{H}}C_H$ 
and moreover:
\begin{itemize}
\item[(SD1)] $\top \equiv \bigvee_{i \in [m]}\bigvee_{P \in \mathcal{P}_i} C_P$;
\item[(SD2)] $\bot \equiv \left( \bigvee_{P \in \mathcal{P}_i} C_P \right) \wedge \left( \bigvee_{P \in \mathcal{P}_j} C_P \right)$ for $i \neq j$ in $[m]$;
\item[(SD3)] $\left( \bigvee_{S \in \mathcal{S}_i} C_S \right) \not\equiv \left( \bigvee_{S \in \mathcal{S}_j} C_S \right)$ for $i \neq j$ in $[m]$.
\end{itemize}

We now use the above development to describe a recursive construction of a canonical SDD for a given Boolean function.  
Let $F$ be a Boolean function and let $T$ be a vtree, both over $X$.  For every node $v \in T$ 
and every subset $\mathcal{H}$ of factors of $F$ relative to $X_v$, 
we construct a  
circuit $C_{v,\mathcal{H}}$, as follows. 

If $v$ is a leaf of $T$, then $X_v=\{x\}$ for some variable $x \in X$.  
There are two cases.  Either $F(0,X\setminus\{x\})=F(1,X\setminus\{x\})$, 
or $F(0,X\setminus\{x\})\neq F(1,X\setminus\{x\})$.  In the former case, 
$\mathsf{factors}(F,\{x\})=\{H\}$ 
and: $C_{v,\emptyset} = \bot$;  $C_{v,\{H\}} = \top$.
In the latter case, $\mathsf{factors}(F,\{x\})=\{H_0,H_1\}$ 
and: 
$C_{v,\emptyset} = \bot$; 
$C_{v,\{H_0\}} = \neg x$; 
$C_{v,\{H_1\}} = x$; 
$C_{v,\{H_0,H_1\}} = \top$. 

If $v$ is a node of $T$ with children $w$ and $w'$, we put
\begin{equation}\label{eq:cvH}
C_{v,\mathcal{H}}=\bigvee_{(\mathcal{P},\mathcal{S}) \in \sd(F,\mathcal{H},Y,Y')} \left( C_{w,\mathcal{P}} \wedge C_{w',\mathcal{S}} \right)\text{.} 
\end{equation}

Finally we put, where $r$ is the root of $T$, 
\begin{equation}\label{eq:sdd-root}
S_{F,T} = C_{r,\{F\}} \text{.}
\end{equation}

\begin{lemma}\label{lemma:sdd-constr-correct}
Let $F$ be a Boolean function and let $T$ be a vtree, both over the variables $X$.  
Let $v \in T$ and let $\mathcal{H} \subseteq \mathsf{factors}(F,X_v)$.  
\begin{itemize}
\item $C_{v,\mathcal{H}}$ is a canonical SDD respecting the vtree $T_v$.
\item $C_{v,\mathcal{H}}$ computes $\bigvee_{H \in \mathcal{H}}H$.
\end{itemize}
\end{lemma}
\begin{proof}
By induction on the depth of $v$ in $T$.  The base case holds by construction. 
The inductive case holds inspection of (\ref{eq:cvH}), 
using (\ref{eq:canonical-sdd-form}) together with (SD1)-(SD3), (\ref{eq:noncanonical-sdd-form}), Lemma~\ref{lemma:union-drc}, 
and the induction hypothesis.
\end{proof}

Therefore $S_{F,T}$ is the canonical SDD computing $F$, 
uniquely determined by the vtree $T$.  We recall the notion of \emph{sentential decision width} 
of $F$ relative to $T$.

\begin{definition}\label{def:sdd-width}
Let $F$ be a Boolean function and let $T$ be a vtree, both over $X$.  
The \emph{sentential decision width of $F$ relative to $T$}, in symbols $\sdw(F,T)$, is defined by  
$$\sdw(F,T)= \max_{v \in T}|\{ g \in S_{F,T} \colon \textup{$g$ is structured by $v$} \}|\text{.}$$
The \emph{SDD width} of $F$, in symbols $\sdw(F)$, is defined by  
$$\sdw(F)= \min\{ \sdw(F,T) \colon \textup{$T$ vtree for $X$} \}\text{.}$$ 
\end{definition}

It is well known that OBDDs are canonical SDDs respecting linear vtrees, 
i.e.\ vtrees where every left child is a leaf \cite{DBLP:conf/ijcai/Darwiche11}; in this case, 
the notion of SDD width in Definition~\ref{def:sdd-width} reduces to the 
usual notion of OBDD width \cite{WegenerBook}.  Moreover, 
as a Boolean function (of $n$ variables) has linear (in $n$) OBDD size 
parameterized by its OBDD width, likewise it has linear SDD size parameterized 
by its SDD width.

\begin{theorem}\label{th:linear-sdd-size}
A Boolean function $F$ of $n$ variables and SDD width $k$ has canonical SDD size $O(kn)$.
\end{theorem}
\begin{proof}
Let $X$ be the variables of $F$, so that $|X|=n$, and let $T$ be a vtree for $X$ 
witnessing SDD width $k$ for $F$.  The circuit $S_{F,T}$ in (\ref{eq:sdd-root}) 
is a canonical SDD by Lemma~\ref{lemma:sdd-constr-correct}.  

Moreover, we claim that $S_{F,T}$ has size $O(kn)$.  The $n$ leaves of $T$ contribute at most $2(n+1)$ input or negation gates in $C_{F,T}$.  
The $n-1$ internal nodes of $T$ contribute each at most $k$ $\wedge$-gates (by the definition of SDD width), 
and each such gate is linked with at most $3$ $\vee$-gates.  Hence $C_{F,T}$ contains at most $2(n+1)+3k(n-1)=O(kn)$ gates.  
\end{proof}

We conclude observing that, for classes of Boolean functions, 
bounded circuit treewidth and bounded SDD width collapse.  Indeed, 
on the one hand, the SDD width of a Boolean function $F$ is bounded above by a computable function 
of its circuit treewidth, namely,
\begin{equation}\label{eq:sdd-ub}
\sdw(F) \leq 2^{2 \cdot \fw(F)+1} \leq 2^{2^{(\ctw(F)+2)2^{\ctw(F)+1}+1}+1} \text{,}
\end{equation}
since in the canonical SDD $S_{F,T}$ for $T$ 
every $\wedge$-gate $g$ structured by a node $v \in T$  
corresponds to a pair of sets of factors of $F$ (plus $\bot$), 
and $\fw(F)$ is bounded above by $\ctw(F)$ as in Lemma~\ref{lemma:jhasuciu}.  
On the other hand, for all Boolean functions $F$, along the lines of Proposition~\ref{prop:fiw-ctw}, 
\begin{equation}\label{eq:sdd-lb}
\ctw(F)/3 \leq \sdw(F)\text{.}
\end{equation}

By combining (\ref{eq:fiw-ub})-(\ref{eq:fiw-lb}) and (\ref{eq:sdd-ub})-(\ref{eq:sdd-lb}), 
the factorized implicant width (resp., SDD width) of a Boolean function is squeezed between computable functions of its SDD width (resp., factorized implicant width).

\section{Query Compilation}\label{sect:part2}

In this section, we show that inversions in unions of 
conjunctive queries, with or without inequalities, 
imply large deterministic structured circuits for their lineages.  

Let $\sigma$ be a relational vocabulary.  
A \emph{union of conjunctive queries (UCQs) with inequalities} 
$Q$ is a disjunction of existentially closed 
conjunctions of atoms $Rx_1 \cdots x_m$ 
and inequalities $x \neq y$, 
where $R \in \sigma$ and $x$, $y$, $x_i$ are variables, 
$i \in [m]$.  We call $Q$ a UCQs if it does not contain inequalities.
The \emph{lineage of a Boolean query $Q$ over a database $D$} 
is a Boolean function $L(Q,D)$ whose Boolean variables 
are the tuples in $D$ such that, for every subdatabase $D' \subseteq D$, 
it holds that $D' \models Q$ iff  $b_{D'} \models L(Q,D)$, 
where $b_{D'} \colon D \to \{0,1\}$ is defined by $b_{D'}(t)=1$ iff $t \in D'$.  
A \emph{lineage of a Boolean query} is a lineage of the query over some database.

\subsection{Inversions Imply Large De\-ter\-mi\-ni\-stic Stru\-ctu\-red Forms}

We prove the main result.  For all $k,n \geq 1$, let $X=\{ x_{l} \colon l \in [n] \}$, 
$Y=\{ y_{m} \colon m \in [n] \}$, 
$Z^i=\{ z^i_{l,m} \colon l,m \in [n] \}$ for $i \in [k]$, 
and $Z=\bigcup_{i \in [k]}Z^i$.  For $i \in [k-1]$, let:
\begin{align*}
H^0_{k,n}(X,Z^1) &= \bigvee_{l,m \in [n]} (x_{l} \wedge z^1_{l,m})\text{,}\\
H^i_{k,n}(Z^i,Z^{i+1}) &= \bigvee_{l,m \in [n]} (z^i_{l,m} \wedge z^{i+1}_{l,m})\text{,}\\
H^k_{k,n}(Z^k,Y) &= \bigvee_{l,m \in [n]} (z^k_{l,m} \wedge y_{m})\text{.}
\end{align*}

In \cite[Proposition~7]{DBLP:journals/mst/JhaS13} and \cite[Theorem~3.9]{DBLP:conf/icdt/JhaS12} 
Jha and Suciu show the following, resp.\ for UCQs and UCQs with inequalities.  

\begin{lemma}\label{lemma:inv-impl-h}
Let $Q$ be a UCQs with or without inequalities.\footnote{As a technical assumption, we assume that all queries and databases are ranked \cite{Suciu:2011:PD:2031527}.}   
If $Q$ \lq\lq contains an inversion of length $k \geq 1$\rq\rq, 
then for every $n \geq 1$ there exist a lineage $F(X)$ of $Q$ on $O(n^2)$ variables and assignments $b_i \colon X_i \to \{0,1\}$ for $X_i \subseteq X$ and $i=0,1,\ldots,k$ 
such that $$F(b_i,X \setminus X_i) \equiv H_{k,n}^i\text{.}$$
\end{lemma}

As we only need the implication stated in Lemma~\ref{lemma:inv-impl-h}, we omit the 
technical definition of the notion of inversion \cite{DBLP:conf/pods/DalviS07a}.   
The following statement unifies and generalizes analogous results 
by Jha and Suciu for UCQs with inequalities vs.\ OBDDs \cite[Theorem~3.9]{DBLP:conf/icdt/JhaS12} and by Beame and Liew 
for UCQs vs.\  SDDs \cite[Theorem~4.6]{DBLP:conf/uai/BeameL15}.

\begin{theorem}\label{th:main-ineq}
Let $Q$ be a UCQs with or without inequalities.  If $Q$ \lq\lq contains an inversion of length $k \geq 1$\rq\rq,  
then for every $n \geq 1$ there exists a lineage $F$ of $Q$ on $O(n^2)$ variables 
whose deterministic structured NNF size is $2^{\Omega(n/k)}$.
\end{theorem}
\begin{proof}[of Theorem~\ref{th:main-ineq}]
Let $Q$ be a query with inequalities.  If $Q$ \lq\lq contains an inversion of length $k \geq 1$\rq\rq, 
then by Lemma~\ref{lemma:inv-impl-h} for every $n \geq 1$ there exist a lineage $F(X)$ of $Q$ on $O(n^2)$ variables and assignments $b_i \colon X_i \to \{0,1\}$ for $X_i \subseteq X$ and $i=0,1,\ldots,k$ 
such that $F(b_i,X \setminus X_i) \equiv H_{k,n}^i$.  

Let $C$ be a deterministic NNF of size $s$ computing $F$, structured by the vtree $T$.  
Then by the properties of deterministic structured NNFs {DBLP:conf/aaai/PipatsrisawatD08}, it holds that 
$C_i(X \setminus X_i)=C(b_i,X \setminus X_i)$ is a deterministic NNF of size $s_i \leq s$ that computes $H_{k,n}^i$ 
and is structured by $T$, for all $i=0,1,\ldots,k$.  By Lemma~\ref{lemma:main-inversion}, 
there exists $i \in \{0,1,\ldots,k\}$ such that $C_i$ has size $s_i=2^{\Omega(n/k)}$.  
Therefore $C$ has size $2^{\Omega(n/k)}$.
\end{proof}

The proof idea is that, if a query $Q$ \lq\lq contains inversions\rq\rq, 
then it has a lineage $L(Q,D)$ of which each $H^i_{k,n}$ is a cofactor ($i=0,1,\ldots,k$).  
If $C$ is a small deterministic form for $L(Q,D)$ respecting a vtree $T$, 
then small deterministic structured forms for each $H^i_{k,n}$, 
\emph{all} respecting the vtree $T$, can be mined from $C$ by suitably assigning its inputs. 
But this is impossible 
for communication complexity reasons (Lemma~\ref{lemma:main-inversion}).

\begin{lemma}\label{lemma:main-inversion}
For every vtree $T$ for $X \cup Y \cup Z$ and every family $\{C_0,\ldots,C_k\}$ of deterministic structured NNFs, 
where $C_i$ is structured by $T$ and computes $H^i_{k,n}$ ($i=0,1,\ldots,k$), 
there exists $i \in \{0,1,\ldots,k\}$ such that $C_i$ has size $2^{\Omega(n/k)}$.
\end{lemma}
\begin{proof}[of Lemma~\ref{lemma:main-inversion}]
We let $X_v$ denote the variables in $T_v \cap X$, 
$Y_v$ denote the variables in $T_v \cap Y$, 
and $Z_v$ denote the variables in $T_v \cap Z$.  

\begin{claim}\label{cl:find-node}
There exists $v \in T$ such that $2n/5 \leq |X_v \cup Y_v| \leq 4n/5$.
\end{claim}
\begin{proof}[Proof of Claim~\ref{cl:find-node}]
Let $n_v=|X_v \cup Y_v|$, for all $v \in T$.  
Let $v_1,\ldots,v_l$ be a root-leaf path in $T$ such that, 
letting $n_i=|X_{v_i} \cup Y_{v_i}|$, 
$$n_{i+1} \geq n_i/2$$
for all $i=1,\ldots,l-1$.  Let $i \in [l]$ be minimum such that 
$$n_i \leq |X \cup Y|/5\text{,}$$
so that 
$$|X \cup Y|/5 < n_{i-1}\text{.}$$
By construction,  
$$n_{i-1} \leq 2n_i \leq 2|X \cup Y|/5\text{.}$$
Hence, letting $v=v_{i-1}$, we have  
$$2n/5=|X \cup Y|/5 < |X_{v} \cup Y_{v}| \leq 2|X \cup Y|/5=4n/5\text{,}$$
and we are done.
\end{proof}

By Claim~\ref{cl:find-node}, let $v \in T$ be such that $2n/5 \leq |X_v \cup Y_v| \leq 4n/5$.  
Let $n_x=|X_v|$ and $n_y=|Y_v|$.  Assume without loss of generality that $n_x \geq n_y$; 
otherwise the argument is similar. It follows by the choice of $v$ that 
\begin{equation}\label{eq:nx-lb}
n_x \geq n/5\text{} 
\end{equation}
and that 
\begin{equation}\label{eq:n-ny-lb}
n-n_y \geq n-n_x \geq n-4n/5=n/5\text{.}
\end{equation}

We enter a case distinction.  The first case is covered by the following claim.

\begin{claim}\label{cl:h0-hard}
If there exists $j \in [n]$ such that for all $i \in [n]$ 
it holds that $x_i \in X_v$ implies $z^1_{i,j} \in T \setminus T_v$, 
then $C_0$ has size $2^{\Omega(n)}$.
\end{claim}

\begin{proof}[Proof of Claim~\ref{cl:h0-hard}]
For $j \in [n]$, let $Z^1_j=\{ z^1_{i,j} \colon x_i \in X_v \} \setminus T_v$.  
By hypothesis, there exists $j \in [n]$ such that $|Z^1_{j}|=n_x$.  Write $$C_0(X,Z^1)=C_0(X_v,X\setminus X_v,Z^1_j,Z^1\setminus Z^1_j)\text{.}$$  
Then 
$$C'_0(X_v,Z^1_j)=C_0(X_v,\{0\}^{X\setminus X_v},Z^1_j,\{0\}^{Z^1\setminus Z^1_j})$$ 
is a deterministic NNF structured by $T$ of size $|C'_0| \leq |C_0|$  \cite{DBLP:conf/aaai/PipatsrisawatD08}.  It follows from 
Theorem~\ref{th:struct} that $C'_0$ has a disjoint rectangle cover of size at most $|C'_0|$ 
where each rectangle has underlying partition $(X_v,Z^1_j)$.

By the choice of $v$, it holds that $X_v$ contains $n_x$ variables in $X$.  
For the sake of notation, say that $X_v=\{x_1,\ldots,x_{n_x}\}$, 
so that $Z^1_j=\{z^1_{1,j},\ldots,z^1_{n_x,j}\}$.  Then, since $C_0 \equiv H^0_{k,n}$, 
we have that 
$$C'_0 \equiv (x_1 \wedge z^1_{1,j}) \vee \cdots \vee (x_{n_x} \wedge z^1_{n_x,j})\text{.}$$
Note that $C'_0(X_v,Z^1_j)$ is the complement of the disjointness function $D_{n_x}(X_v,Z^1_j)$ in (\ref{eq:df}).  
Therefore the complement of the communication matrix of $C'_0$ relative to $(X_v,Z^1_j)$ 
is equal (up to a permutation of rows and columns) to the communication matrix of $D_{n_x}$, i.e., 
$$\mathsf{cm}(D_{n_x},X_v,Z^1_j)=1-\mathsf{cm}(C'_0,X_v,Z^1_j)\text{,}$$
where $1$ denotes the $2^{n_x} \times 2^{n_x}$ all-$1$ matrix.  Therefore, 
by (\ref{eq:rkdf}) and basic linear algebra,
\begin{align*}
2^{n_x} &= \mathsf{rank}(\mathsf{cm}(D_{n_x},X_v,Z^1_j))\\
        &= \mathsf{rank}(1-\mathsf{cm}(C'_0,X_v,Z^1_j))\\
        &\leq \mathsf{rank}(1)+\mathsf{rank}(\mathsf{cm}(C'_0,X_v,Z^1_j))\\
        &=1+\mathsf{rank}(\mathsf{cm}(C'_0,X_v,Z^1_j))\text{,}
\end{align*}
hence
\begin{equation}\label{eq:lb-compl-df}
\mathsf{rank}(\mathsf{cm}(C'_0,X_v,Z^1_j)) \geq 2^{n_x}-1\text{.} 
\end{equation}
Therefore, by Theorem~\ref{th:rklb} and (\ref{eq:lb-compl-df}), 
every disjoint rectangle cover of $C'_0$ into rectangles with underlying 
partition $(X_v,Z^1_j)$ contains at least $2^{n_x}-1$ rectangles.

Summarizing, $C'_0$ has a disjoint rectangle cover of size at most $|C'_0|$ 
where each rectangle has underlying partition $(X_v,Z^1_j)$, 
but every such rectangle cover contains at least $2^{n_x}-1$ rectangles.  
Hence $$|C_0| \geq |C'_0| \geq 2^{n_x}-1 \geq 2^{n/5}-1=2^{\Omega(n)}\text{,}$$
and we are done (recall (\ref{eq:nx-lb})).
\end{proof}

The second (and complementary) case is covered by the following claim.

\begin{claim}\label{cl:hkp-hard}
If for all $j \in [n]$ there exists $i \in [n]$ 
such that $x_i \in X_v$ and $z^1_{i,j} \in T_v$, 
then there exists $p \in [k]$ such that $C_p$ has size $2^{\Omega(n/k)}$.
\end{claim}

\begin{proof}[Proof of Claim~\ref{cl:hkp-hard}]  
Define a set $S$ 
of pairs $(i,j)$ as follows.  For each $j \in [n]$ such that $y_j \in T\setminus T_v$, 
choose $i \in [n]$ such that $z^1_{i,j} \in T_v$, and add $(i,j)$ to $S$.  Note that 
\begin{equation}\label{eq:s-y-rel}
|S|=n-n_y\text{.} 
\end{equation}

For each $p=1,\ldots,k-1$, let $R_p \subseteq S$ be such that $(i,j) \in R_p$ iff 
$z^1_{i,j},\ldots,z^p_{i,j} \in T_v$ and $z^{p+1}_{i,j} \in T\setminus T_v$.  Also, 
let $$R_k=S\setminus \bigcup_{i=1}^{k-1}R_i\text{.}$$  Note that $R_1,\ldots,R_k$ form a partition of $S$, 
so that 
\begin{equation}\label{eq:r-part}
|R_1|+\cdots+|R_k|=|S|\text{.} 
\end{equation}

We show that $$\sum_{p=1}^k |C_p| \geq k(2^{n/5k}-1)\text{,}$$
which implies that there exists $p \in [k]$ such that $|C_p| \geq 2^{n/5k}-1$, and we are done.

Let $p \in [k-1]$.  For all $(i,j) \in R_p$, it holds that 
$z^p_{i,j} \in T_v$ and $z^{p+1}_{i,j} \in T\setminus T_v$.  
Let 
$V^p=\{ z^p_{i,j} \colon (i,j)\in R_p \}$ 
and 
$V^{p+1}=\{ z^{p+1}_{i,j} \colon (i,j)\in R_p \}$.  Write
$$C_p(Z^p,Z^{p+1})=C_p(V^p,Z^p\setminus V^p, V^{p+1}, Z^{p+1}\setminus V^{p+1})$$
and let 
$$C'_p(V^p,V^{p+1})=C_p(V^p,\{0\}^{Z^p\setminus V^p}, V^{p+1}, \{0\}^{Z^{p+1}\setminus V^{p+1}})\text{,}$$
so that $C'_p$ is a deterministic NNF structured by $T$ of size $|C'_p| \leq |C_p|$  \cite{DBLP:conf/aaai/PipatsrisawatD08}.  It follows from 
Theorem~\ref{th:struct} that $C'_p$ has a disjoint rectangle cover of size at most $|C'_p|$ 
where each rectangle has underlying partition $(V^p,V^{p+1})$.

Since $C_p \equiv H^p_{k,n}$, 
we have that 
$$C'_p \equiv \bigvee_{(i,j) \in R_p} (z^p_{i,j} \wedge z^{p+1}_{i,j})\text{,}$$
and along the lines of Claim~\ref{cl:h0-hard} we obtain 
\begin{equation}\label{eq:r-p-l}
|C_p| \geq |C'_p| \geq 2^{|R_p|}-1 \text{.} 
\end{equation}
Similarly, we obtain 
\begin{equation}\label{eq:r-k-l}
|C_k| \geq 2^{|R_k|}-1 \text{.}
\end{equation}

Thus 
\begin{align*}
\sum_{p=1}^k |C_p| &\geq \sum_{p=1}^k (2^{|R_p|}-1) & \text{(\ref{eq:r-p-l}), (\ref{eq:r-k-l})}\\
                 &\geq \sum_{p=1}^k 2^{|R_p|}-k & \\
                 & \geq k2^{\sum_{p=1}^k|R_p|/k}-k & \\
                 &=k2^{|S|/k}-k & \text{(\ref{eq:r-part})} \\
                 &\geq k2^{n/5k}-k & \text{(\ref{eq:n-ny-lb}), (\ref{eq:s-y-rel})} 
\end{align*}
and we are done.  The third inequality holds by plugging 
the convex function $f(r)=2^r$ in Jensen's inequality $\sum_{p=1}^k f(r_p)/k \geq f(\sum_{p=1}^k r_p/k)$.
\end{proof}

Claim~\ref{cl:h0-hard} and Claim~\ref{cl:hkp-hard} imply the statement. 
\end{proof}

\section{Conclusion}\label{sect:concl}

We have related the circuit treewidth of a Boolean 
function with the width of its SDD implementation (and more generally 
its width in natural canonical classes of deterministic structured forms), 
and we have incorporated constant width 
SDDs and polynomial size SDDs in the panorama of query compilation 
for union of conjunctive queries with and without negations. 

The comparison of Theorem~\ref{th:main-ineq} and \cite[Theorem~4.6]{DBLP:conf/uai/BeameL15} 
reiterates the question about the relative succinctness of deterministic structured forms 
and SDDs \cite{DBLP:conf/uai/BeameL15,DBLP:conf/aaai/Bova16}.  As Beame and Liew observe, a natural candidate function 
for an exponential separation is the \emph{indirect access storage} (\emph{ISA}) function, 
which is known to have large OBDDs but whose small deterministic structured forms 
deviate substantially from the SDD syntax \cite[Section~6]{DBLP:conf/uai/BeameL15}.  However, 
as we prove in Appendix~\ref{app:isaproof}, 
ISA has small SDD size,  
which unfortunately leaves us with no candidates for a separation.

The canonical structured deterministic forms 
induced by factorized implicants, introduced in Section~\ref{sect:fimpldec}, 
deserve in our opinion both a direct investigation in the framework of the knowledge compilation map \cite{DBLP:journals/jair/DarwicheM02}, 
and a thorough comparison with the data structures used in \emph{factorized databases}, 
which are more than just reminiscent of structured deterministic forms \cite{Olteanu16}.

The question remains whether bounded circuit tree\-width lineages imply bounded OBDD width 
for UCQs with inequalities, 
as conjectured by Jha and Suciu \cite{DBLP:conf/icdt/JhaS12}.  

Another intriguing conjecture is that SDDs with OR gates of 
bounded fanin are quasipolynomially simulated by OBDDs.\footnote{Personal communication with Igor Razgon.} The containment of bounded width SDDs in polynomial size OBDDs, discussed in the introduction 
and obtained in Section~\ref{sect:part1}, imply a 
polynomial simulation of bounded width SDDs, 
which have indeed bounded fanin ORs, by OBDDs.  

\section*{Acknowledgments} This research was supported by the FWF Austrian Science Fund (Parameterized Compilation, P26200).

\appendix

\section{ISA Has Small SDD Size}\label{app:isaproof}

Let $k$ and $m$ be positive integers such that $2^k m=2^m$.   
The \emph{indirect access storage} (in short, \emph{ISA}) function on $n=k+2^km=k+2^m$ variables 
$$\mathrm{ISA}_n(y_1,\ldots,y_{k},x_{1,1},\ldots,x_{1,m},\ldots,x_{2^k,1},\ldots,x_{2^k,m})$$
also displayed as 
$$\mathrm{ISA}_n(Y_k,Z_m)=\mathrm{ISA}_n(y_1,\ldots,y_{k},z_{1},\ldots,z_{2^m})$$
accepts input 
$$a_1,\ldots,a_{k},b_{1,1},\ldots,b_{1,m},\ldots,b_{2^k,1},\ldots,b_{2^k,m}$$
also displayed as 
$$a_1,\ldots,a_{k},c_{1},\ldots,c_{2^m}$$
iff , 
letting $i-1 \in \{0,\ldots,2^k-1\}$ be the number 
whose binary representation is $(a_1,\ldots,a_{k})$ 
and $j-1 \in \{0,\ldots,2^m-1\}$ be the number 
whose binary representation is $(b_{i,1},\ldots,b_{i,m})$, 
it holds that $c_{j}=1$.

\begin{proposition}
$\mathrm{ISA}_n$ has SDD size $O(n^{13/5})$.
\end{proposition}
\begin{proof}
Let $T_n=T(Y_k,Z_m)$ be the vtree for variables $Y_k \cup Z_m$ 
formed by a right-linear subtree $T_n(Y_k)$ whose left leaves correspond to the variables in $Y_k$ 
and whose (unique) right leaf $v$ is the root of a left-linear subtree $T_n(Z_m)$ 
whose (unique) left leaf corresponds to $z_{1}$ and whose right leaves correspond (in a postorder traversal) to $z_2,\ldots,z_{2^m}$. 
For instance, the vtree $T_5(Y_1,Z_2)$ is depicted in Figure~\ref{fig:isaproof}.  

\begin{figure}[t]
\begin{center}
\begin{picture}(0,0)%
\includegraphics{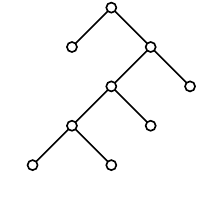}%
\end{picture}%
\setlength{\unitlength}{4144sp}%
\begingroup\makeatletter\ifx\SetFigFont\undefined%
\gdef\SetFigFont#1#2#3#4#5{%
  \reset@font\fontsize{#1}{#2pt}%
  \fontfamily{#3}\fontseries{#4}\fontshape{#5}%
  \selectfont}%
\fi\endgroup%
\begin{picture}(930,901)(11551,-9931)
\put(11566,-9871){\makebox(0,0)[lb]{\smash{{\SetFigFont{7}{8.4}{\familydefault}{\mddefault}{\updefault}{\color[rgb]{0,0,0}$z_1$}%
}}}}
\put(12106,-9871){\makebox(0,0)[lb]{\smash{{\SetFigFont{7}{8.4}{\familydefault}{\mddefault}{\updefault}{\color[rgb]{0,0,0}$z_2$}%
}}}}
\put(12286,-9691){\makebox(0,0)[lb]{\smash{{\SetFigFont{7}{8.4}{\familydefault}{\mddefault}{\updefault}{\color[rgb]{0,0,0}$z_3$}%
}}}}
\put(12466,-9511){\makebox(0,0)[lb]{\smash{{\SetFigFont{7}{8.4}{\familydefault}{\mddefault}{\updefault}{\color[rgb]{0,0,0}$z_4$}%
}}}}
\put(12286,-9196){\makebox(0,0)[lb]{\smash{{\SetFigFont{7}{8.4}{\familydefault}{\mddefault}{\updefault}{\color[rgb]{0,0,0}$v_4$}%
}}}}
\put(11791,-9331){\makebox(0,0)[lb]{\smash{{\SetFigFont{7}{8.4}{\familydefault}{\mddefault}{\updefault}{\color[rgb]{0,0,0}$y_1$}%
}}}}
\end{picture}%
\end{center}
\caption{Vtree for $\mathrm{ISA}_5$.}
\label{fig:isaproof}
\end{figure} 

We liberally identify the leaves of $T_n$ with their labels (so that we simply call $x$ 
the leaf of $T_n$ labelled by the variable $x$).  For $i \in [k]$, we let $w_i$ 
denote the node in $T_n$ whose left child is $y_i$; for $j \in [2^m]$, we let $v_j$ 
denote the node in $T_n$ whose right child is $z_j$.  

A term on $Z_m$ is a conjunction of literals on $Z_m$.\footnote{The empty term is denoted by $\top$ 
and a term containing both literals of a variable is denoted by $\bot$.}  
We call a term on $Z_m$ \emph{small} if it contains at most $m+1$ distinct variables.  
Note that the number of small terms on $Z_m$ is
\begin{equation}\label{eq:smallterms}
3^{m+1}+1=O(n^{8/5}) 
\end{equation}
since $m = \log_2(2^m)=\log_2(n-k) \leq \log_2 n$.  

We now construct an SDD $C$ computing $\mathrm{ISA}_n$ and respecting $T_n$, 
where each $\wedge$-gate structured by a node of the form $v_j$ ($j \in [2^m]$)
conjoins one small term on $Z_m$ and an input gate (namely, a literal on $v_j$ or a constant).  
As $C$ has at most $2n+2=O(n)$ input gates, 
it follows that the number of $\wedge$-gates in $C$ 
structured by nodes of the form $v_{j}$ ($j \in [2^m]$) is $$O(n^{13/5})\text{.}$$ 
Moreover, $C$ is such that the number of $\wedge$-gates 
structured by nodes of the form $w_i$ ($i \in [k]$) is
$$2^{k+1}-2=O(n)$$ 
since $k \leq \log_2 n$.  It follows immediately  
that $C$ has size polynomial in $n$.  
Indeed, as each $\wedge$-gate in $C$ contributes a constant number of $\vee$-gates, 
$C$ contains at most $O(n^{13/5})$ internal gates.  
Also, $C$ has at most $2n+2=O(n)$ input gates.  Hence $$|C|=O(n^{13/5})\text{.}$$

We now present the construction of the SDD $C$ implementing $\mathrm{ISA}_n$ and respecting $T_n$.  

The upper part of $C$ is isomorphic to an OBDD respecting the order $y_1<\cdots<y_k$ 
and having $2^k$ source gates.  Each such source, say $g_{a_1,\ldots,a_k}$, 
corresponds to the Boolean assignment $y_i \mapsto a_i$, $i \in [k]$,  
of the variables in $Y_k$ and implements the cofactor 
\begin{equation}\label{eq:isacof}
\mathrm{ISA}_n(a_1,\ldots,a_{k},z_{1},\ldots,z_{2^m}) 
\end{equation}
as an SDD respecting the vtree $T_n(Z_m)$, as follows.

We start observing that each cofactor in (\ref{eq:isacof}) is expressible as a sentential decision, respecting the root node 
of $T_n(Z_m)$ and involving only small terms on $Z_m$ (and literals on $z_{2^m}$).   

\begin{claim}\label{cl:base}
The function $\mathrm{ISA}_n(a_1,\ldots,a_{k},z_{1},\ldots,z_{2^m})$ 
is equivalent to a sentential decision $\bigvee_i(P_i \wedge S_i)$ of the form (\ref{eq:sdd-def}), 
where the $\wedge$-gates are structured by $v_{2^m} \in T_n$ 
and the $P_i$'s are small terms.
\end{claim}
\begin{proof}[of Claim~\ref{cl:base}]
We distinguish two cases.  

\medskip\noindent Case $a_1+\cdots+a_{k}=k$:  In this case, we have to implement $\mathrm{ISA}_n(1,\ldots,1,z_{1},\ldots,z_{2^m})$, 
which is equivalent to 
$$\bigvee_{j=1}^{2^m}( \textup{\lq\lq$x_{2^k,1},\ldots,x_{2^k,m} = j$\rq\rq} \wedge z_j)\text{;}$$
here $\textup{\lq\lq$x_{2^k,1},\ldots,x_{2^k,m} = j$\rq\rq}$ 
corresponds to the term 
$$L^{a_1}_{2^k,1} \wedge \cdots \wedge L^{a_m}_{2^k,m}$$ 
where $a_1\cdots a_m$ represents $j-1 \in \{0,\ldots,2^m-1\}$ in binary, 
and $L^{0}_{2^k,j'}=\neg x_{2^k,j'}$, $L^{1}_{2^k,j'}=x_{2^k,j'}$, $j' \in [m]$.  

For all $a \colon \{x_{2^k,1},\ldots,x_{2^k,m-1}\} \to \{0,1\}$ 
and $i \in [2^{m-1}]$ we say that \emph{$a$ orbits on $i$} 
if $$a(x_{2^k,1})\cdots a(x_{2^k,m-1})1$$ represents $2i-1$ in binary.  
If $a$ orbits on $i$, we let \lq\lq $a$ orbits on $i$\rq\rq\ denote the term 
$$L^{a(x_{2^k,1})}_{2^k,1} \wedge \cdots \wedge L^{a(x_{2^k,m-1})}_{2^k,m-1}\text{.}$$ 
By direct inspection,  
$\mathrm{ISA}_n(1,\ldots,1,z_{1},\ldots,z_{2^m})$ is equivalent to a sentential decision $\bigvee_i(P_i \underset{*}{\wedge} S_i)$ as in (\ref{eq:sdd-def}), namely,  
\begin{align*}
\bigvee_{(a,i)} 
\begin{cases} 
\left(\textup{\lq\lq$a$ orbits on $i$\rq\rq} \wedge \neg z_{2i-1} \wedge \neg z_{2i} \right) \underset{*}{\wedge} \bot \\ 
\left(\textup{\lq\lq$a$ orbits on $i$\rq\rq} \wedge \neg z_{2i-1} \wedge z_{2i} \right) \underset{*}{\wedge} z_{2^m} \\
\left(\textup{\lq\lq$a$ orbits on $i$\rq\rq} \wedge z_{2i-1} \wedge \neg z_{2i} \right) \underset{*}{\wedge} \neg z_{2^m} \\
\left(\textup{\lq\lq$a$ orbits on $i$\rq\rq} \wedge z_{2i-1} \wedge z_{2i} \right) \underset{*}{\wedge} \bot 
\end{cases}
\end{align*}
where $(a,i)$ ranges over all pairs such that $a$ orbits on $i$.  

Here, the interesting $\wedge$-gates (marked with $*$) are structured by the node $v_{2^m}$ in $T_n$.  
Moreover, the $P_i$'s are small terms as they contain $(m-1)+2=m+1$ variables by construction.

\begin{example}[$k=2$, $m=4$]\label{ex:hec} Assume $y_1=y_2=1$, so that we compute $\mathrm{ISA}_{18}(1,1,z_1,\ldots,z_{16})$.  
In this case $a \colon \{z_{13},z_{14},z_{15}\} \to \{0,1\}$ and $i \in \{0,1,\ldots,8\}$. The following lists 
the disjuncts corresponding to $\textup{\lq\lq$a$ orbits on $i$\rq\rq}$ for $i=0,1,2,3,4,6$ (we use $\overline{x}=\neg x$ 
and $xy=x \wedge y$ as shortenings):   
\begin{align*}
\bigvee 
\begin{cases} 
\overline{z_{13}}\overline{z_{14}}\overline{z_{15}}\overline{z_1}\overline{z_2} \underset{*}{\wedge} \bot \\ 
\overline{z_{13}}\overline{z_{14}}\overline{z_{15}}\overline{z_1}z_2 \underset{*}{\wedge} z_{16}\\ 
\overline{z_{13}}\overline{z_{14}}\overline{z_{15}}z_1\overline{z_2} \underset{*}{\wedge} \overline{z_{16}} \\ 
\overline{z_{13}}\overline{z_{14}}\overline{z_{15}}z_1 z_2 \underset{*}{\wedge} \top \\ 
\overline{z_{13}}\overline{z_{14}}z_{15}\overline{z_3}\overline{z_4} \underset{*}{\wedge} \bot \\ 
\overline{z_{13}}\overline{z_{14}}z_{15}\overline{z_3}z_4 \underset{*}{\wedge} z_{16} \\ 
\overline{z_{13}}\overline{z_{14}}z_{15}z_3\overline{z_4} \underset{*}{\wedge} \overline{z_{16}} \\ 
\overline{z_{13}}\overline{z_{14}}z_{15}z_3 z_4 \underset{*}{\wedge} \top \\ 
\vdots \\
z_{13}\overline{z_{14}}z_{15}\overline{z_{11}}\overline{z_{12}} \underset{*}{\wedge} \bot \\ 
z_{13}\overline{z_{14}}z_{15}\overline{z_{11}}z_{12} \underset{*}{\wedge} z_{16} \\ 
z_{13}\overline{z_{14}}z_{15}z_{11}\overline{z_{12}} \underset{*}{\wedge} \overline{z_{16}} \\ 
z_{13}\overline{z_{14}}z_{15}z_{11}z_{12} \underset{*}{\wedge} \top
\end{cases}
\end{align*} 
\end{example}

If $(a,i)$ gives $\{ z_{2i-1},  z_{2i} \} \subseteq \{x_{2^k,1},\ldots,x_{2^k,m-2}\}$, 
then the corresponding subdisjunction reduces to one disjunct only; 
and, if $(a,i)$ gives $\{ z_{2i-1},  z_{2i} \}=\{x_{2^k,m-1},x_{2^k,m}\}$, 
then the corresponding subdisjunction reduces to two disjuncts only, as the following example illustrates.

\begin{example}[$k=2$, $m=4$] Continuing Example~\ref{ex:hec}, the following lists 
the one disjunct corresponding to $\textup{\lq\lq$a$ orbits on $7$\rq\rq}$:   
\begin{align*}
z_{13}z_{14}\overline{z_{15}} \underset{*}{\wedge} \top
\end{align*} 
and the following lists the two disjuncts corresponding to $\textup{\lq\lq$a$ orbits on $8$\rq\rq}$:   
\begin{align*}
\bigvee 
\begin{cases} 
z_{13}z_{14}z_{15} \underset{*}{\wedge} \overline{z_{16}} \\ 
z_{13}z_{14}z_{15} \underset{*}{\wedge} z_{16}
\end{cases}
\end{align*} 
\end{example}

\medskip\noindent Case $a_1+\cdots+a_{k}<k$: Say that $a_1\cdots a_{k}$ represents $i<2^k-1$.  We implement $\mathrm{ISA}_n(a_1,\ldots,a_k,z_{1},\ldots,z_{2^m})$, 
which is equivalent to 
$$\bigvee_{j=1}^{2^m}( \textup{\lq\lq$x_{i,1},\ldots,x_{i,m} = j$\rq\rq} \wedge z_j)\text{,}$$
where the notation is as in the previous case; note that $z_{2^m} \not\in \{x_{i,1},\ldots,x_{i,m}\}$.  An equivalent sentential decision of the form (\ref{eq:sdd-def}) 
is obtained by disjoining 
$$\textup{\lq\lq$x_{i,1},\ldots,x_{i,m} = 2^m$\rq\rq} \underset{*}{\wedge} z_{2^m}$$
and the following:
\begin{align*}
\bigvee_{j=1}^{2^m-1} 
\begin{cases} 
\left(\textup{\lq\lq$x_{i,1},\ldots,x_{i,m} = j$\rq\rq} \wedge \neg z_j \right) \underset{*}{\wedge} \bot \\ 
\left(\textup{\lq\lq$x_{i,1},\ldots,x_{i,m} = j$\rq\rq} \wedge z_j \right) \underset{*}{\wedge} \top 
\end{cases}
\end{align*}
where the interesting $\wedge$-gates (marked with $*$) are structured by the node $v_{2^m}$ in $T_n$.  
Moreover, the $P_i$'s are small terms as they contain $m+1$ variables by construction.

If $\textup{\lq\lq$x_{i,1},\ldots,x_{i,m} = j$\rq\rq}$ and $z_j \in \{x_{i,1},\ldots,x_{i,m}\}$, 
then the corresponding pair of disjuncts 
simplifies, as the following example illustrates.

\begin{example}[$k=2$, $m=4$] Assume $y_1=0$ and $y_2=1$, 
so that we compute $\mathrm{ISA}_{18}(0,1,z_1,\ldots,z_{16})$ by the following sentential decision:   
\begin{align*}
\bigvee 
\begin{cases} 
(\textup{\lq\lq$z_5,\ldots,z_8 = 1$\rq\rq} \wedge \neg z_1) \underset{*}{\wedge} \bot \\
(\textup{\lq\lq$z_5,\ldots,z_8 = 1$\rq\rq} \wedge z_1) \underset{*}{\wedge} \top \\
\vdots \\
(\textup{\lq\lq$z_5,\ldots,z_8 = 4$\rq\rq} \wedge \neg z_4) \underset{*}{\wedge} \bot \\
(\textup{\lq\lq$z_5,\ldots,z_8 = 4$\rq\rq} \wedge z_4) \underset{*}{\wedge} \top \\
\overline{z_5}z_6 \overline{z_7} \overline{z_8} \underset{*}{\wedge} \bot \\
\overline{z_5}z_6 \overline{z_7} z_8 \underset{*}{\wedge} \top \\
\overline{z_5}z_6 z_7 \overline{z_8} \underset{*}{\wedge} \top \\
\overline{z_5}z_6 z_7 z_8 \underset{*}{\wedge} \top \\
(\textup{\lq\lq$z_5,\ldots,z_8 = 9$\rq\rq} \wedge \neg z_1) \underset{*}{\wedge} \bot \\
(\textup{\lq\lq$z_5,\ldots,z_8 = 9$\rq\rq} \wedge z_1) \underset{*}{\wedge} \top \\
\vdots \\
(\textup{\lq\lq$z_5,\ldots,z_8 = 15$\rq\rq} \wedge \neg z_{15}) \underset{*}{\wedge} \bot \\
(\textup{\lq\lq$z_5,\ldots,z_8 = 15$\rq\rq} \wedge z_{15}) \underset{*}{\wedge} \top \\
(\textup{\lq\lq$z_5,\ldots,z_8 = 16$\rq\rq}) \underset{*}{\wedge} z_{16}
\end{cases}
\end{align*} 
\end{example}

The claim is settled.\end{proof}

The construction implements each gate $g_{a_1,\ldots,a_k}$ 
by the sentential decision given by Claim~\ref{cl:base}.  We now claim 
that the construction can continue recursively by implementing 
the resulting $P_i$'s as SDDs respecting subtrees of the subtree of $T_n$ rooted at $v_{2^m-1}$.

\begin{claim}\label{cl:isarec}
Let $P$ be a small term on $Z_m$.
Then $P$ is equivalent to a sentential decision $\bigvee_i(P_i \wedge S_i)$ of the form (\ref{eq:sdd-def}), 
where the $\wedge$-gates are structured by some $v_j \in T_n$ ($j \in [2^m]$)
and the $P_i$'s are small terms.
\end{claim}
\begin{proof}[of Claim~\ref{cl:isarec}]
For $j \in [2^m]$, let $L^0_{j}= \neg z_j$, $L^1_{j}=z_j$, and $L_j \in \{L^c_{j} \colon c=0,1 \}$.  
Say that $P=P(z_{j_1},\ldots,z_{j_{l}})$ where $j_1<\cdots<j_{l-1}<j_l$ and $1<l\leq m+1$.

Let $a$ be the assignment of $\{z_{j_1},\ldots,z_{j_{l-1}}\}$ to $\{0,1\}$ such that 
$$P=\left( \left( L^{a(z_{j_1})}_{j_1} \wedge \cdots \wedge L^{a(z_{j_{l-1}})}_{j_{l-1}} \right) \underset{*}{\wedge} L_{j_l} \right)\text{.}$$
By direct inspection, $P$ is equivalent to the sentential decision $\bigvee_i(P_i \underset{*}{\wedge} S_i)$ of the form (\ref{eq:sdd-def}) 
$$P \vee \bigvee_{a \neq b}\left( \left( L^{b(z_{j_1})}_{j_1} \wedge \cdots \wedge L^{b(z_{j_{l-1}})}_{j_{l-1}} \right) \underset{*}{\wedge} \bot \right)$$
where $b$ ranges over the assignments of $\{z_{j_1},\ldots,z_{j_{l-1}}\}$ in $\{0,1\}$ distinct from $a$.

Here, the interesting $\wedge$-gates (marked with $*$) are structured by the node $v_{j_l}$ in $T_n$.  
Moreover, the $P_i$'s are trivially small terms, because $l-1 \leq l \leq m+1$ by hypothesis.
\end{proof}
The statement is proved.
\end{proof}

\end{document}